\documentclass{ws-ijmpa}

\usepackage[active]{srcltx}
\usepackage{mathrsfs}
\usepackage{bm}
\usepackage{hyperref}
\usepackage[super]{cite}
\usepackage{enumerate}
\usepackage{mathtools}
\usepackage{color}
\usepackage{framed}
\usepackage[active]{srcltx}
\definecolor{shadecolor}{rgb}{.8,.8,.8}
\usepackage[matrix,frame,arrow]{xy}
\usepackage{amsmath}

\newcommand{\ket}[1]{\left\vert{#1}\right\rangle}
\newcommand{\qw}[1][-1]{\ar @{-} [0,#1]}

\newcommand{\gate}[1]{*{\xy *+<.6em>{#1};p\save+LU;+RU **\dir{-}\restore\save+RU;+RD **\dir{-}\restore\save+RD;+LD **\dir{-}\restore\POS+LD;+LU **\dir{-}\endxy} \qw}

\newcommand{\measureD}[1]{*{\xy*+=+<.5em>{\vphantom{\rule{0em}{.1em}#1}}*\cir{r_l};p\save*!R{#1} \restore\save+UC;+UC-<.5em,0em>*!R{\hphantom{#1}}+L **\dir{-} \restore\save+DC;+DC-<.5em,0em>*!R{\hphantom{#1}}+L **\dir{-} \restore\POS+UC-<.5em,0em>*!R{\hphantom{#1}}+L;+DC-<.5em,0em>*!R{\hphantom{#1}}+L **\dir{-} \endxy} \qw}

\newcommand{\multimeasureD}[2]{*+<1em,.9em>{\hphantom{#2}}\save[0,0].[#1,0];p\save !C *{#2},p+LU+<0em,0em>;+RU+<-.8em,0em> **\dir{-}\restore\save +LD;+LU **\dir{-}\restore\save +LD;+RD-<.8em,0em> **\dir{-} \restore\save +RD+<0em,.8em>;+RU-<0em,.8em> **\dir{-} \restore \POS !UR*!UR{\cir<.9em>{r_d}};!DR*!DR{\cir<.9em>{d_l}}\restore \qw}

\newcommand{\qswap}{*=<0em>{\times} \qw}
\newcommand{\multigate}[2]{*+<1em,.9em>{\hphantom{#2}} \qw \POS[0,0].[#1,0];p !C *{#2},p \save+LU;+RU **\dir{-}\restore\save+RU;+RD **\dir{-}\restore\save+RD;+LD **\dir{-}\restore\save+LD;+LU **\dir{-}\restore}
\newcommand{\ghost}[1]{*+<1em,.9em>{\hphantom{#1}} \qw}
\newcommand{\Qcircuit}[1][0em]{\xymatrix @*=<#1>}

\newcommand{\pureghost}[1]{*+<1em,.9em>{\hphantom{#1}}}
\newcommand{\multiprepareC}[2]{*+<1em,.9em>{\hphantom{#2}}\save[0,0].[#1,0];p\save !C
  *{#2},p+RU+<0em,0em>;+LU+<+.8em,0em> **\dir{-}\restore\save +RD;+RU **\dir{-}\restore\save
  +RD;+LD+<.8em,0em> **\dir{-} \restore\save +LD+<0em,.8em>;+LU-<0em,.8em> **\dir{-} \restore \POS
  !UL*!UL{\cir<.9em>{u_r}};!DL*!DL{\cir<.9em>{l_u}}\restore}
\newcommand{\prepareC}[1]{*{\xy*+=+<.5em>{\vphantom{#1\rule{0em}{.1em}}}*\cir{l^r};p\save*!L{#1} \restore\save+UC;+UC+<.5em,0em>*!L{\hphantom{#1}}+R **\dir{-} \restore\save+DC;+DC+<.5em,0em>*!L{\hphantom{#1}}+R **\dir{-} \restore\POS+UC+<.5em,0em>*!L{\hphantom{#1}}+R;+DC+<.5em,0em>*!L{\hphantom{#1}}+R **\dir{-} \endxy}}

\newcommand{\poloFantasmaCn}[1]{{{}^{#1}_{\phantom{#1}}}}

\usepackage{xargs,xifthen,xspace}

\newcommand\Ket[1]{|#1\rangle}

\newcommandx\KetBra[2][usedefault, addprefix=\global, 2=]{|#1\rangle\!\langle\ifthenelse{\equal{#2}{}}{#1}{#2}|}

\newcommand\State[1]{|#1)}

\newcommand\Effect[1]{(#1|}

\newcommand{\RBraKet}[2]{(#1|#2)}

\newcommand\SetStates{\mathrm{St}}

\newcommand\SetEffects{\mathrm{Eff}}

\newcommand\SetTransf{\mathrm{Transf}}

\newcommand\Tr{\operatorname{Tr}}

\newcommand\Test[1]{\bm{\mathcal #1}}

\newcommand\PreparationTest[1]{\bm{#1}}

\newcommand\ObservationTest[1]{\bm{#1}}

\newcommand\Transformation[1]{\mathcal{#1}}

\newcommandx\STransformation[5][usedefault, addprefix=\global, 2=, 3=, 4=,
5=]{\Transformation{#1}_{#2\ifthenelse{\NOT\isempty{#3}}{\,#3}{}}\ifthenelse{
\NOT\isempty{#4}}{^{#4   
\ifthenelse{ \NOT\isempty { #5 } } {,{#5}}{}
}}{}}

\newcommandx\NSState[5][usedefault, addprefix=\global, 2=, 3=, 4=,
5=]{#1_{#2\ifthenelse{\NOT\isempty{#3}}{\,#3}{}#4\ifthenelse{\NOT\isempty{#5}}{,{#5}}
{}}}

\newcommandx\NSEffect[5][usedefault, addprefix=\global, 2=, 3=, 4=,
5=]{#1_{#2\ifthenelse{\NOT\isempty{#3}}{\,#3}{}#4\ifthenelse{\NOT\isempty{#5}}{,#5}{}}}

\newcommandx\SState[5][usedefault, addprefix=\global, 2=, 3=, 4=,
5=]{\State{\NSState{#1}[#2][#3][#4][#5]}}

\newcommandx\SEffect[5][usedefault, addprefix=\global, 2=, 3=, 4=, 5=]{
\Effect{\NSEffect{#1}[#2][#3][#4][#5]} }

\newcommandx\ProbCond[4][usedefault, addprefix=\global, 1=,
2=]{\operatorname{Pr}_{#1}^{#2}\![#3|#4]}

\newcommand{\ustickcool}[1]{*!D!<0em,-.1em>=<0em>{\scriptstyle #1}}

\newcommand{\System}[1]{\mathrm{#1}}

\newcommand\Hilb[1]{\mathscr{#1}}

\newcommand\Herm{\operatorname{Herm}}

\newcommand\Reals{\mathbb R}
\newcommand\Complexes{\mathbb C}

\newcommandx\F[1][usedefault, addprefix=\global, 1=]{\varphi\ifthenelse{\NOT\isempty{#1}}{_{#1}}{}}

\newcommand{\sA}{\System{A}}
\newcommand{\sB}{\System{B}}
\newcommand{\sC}{\System{C}}
\newcommand{\sD}{\System{D}}
\newcommand{\sE}{\System{E}}
\newcommand{\sF}{\System{F}}
\newcommand{\sI}{\System{I}}

\newcommand{\sQ}{\System{Q}}
\newcommand{\sX}{\System{X}}
\newcommand{\sY}{\System{Y}}

\newcommand{\set}[1]{\mathsf{#1}}

\newcommand\Opt{\textsc{opt}\xspace}
\newcommand\Qt{\textsc{qt}\xspace}
\newcommand\Rqt{\textsc{rqt}\xspace}
\newcommand\Fqt{\textsc{fqt}\xspace}

\newcommand\Povm{\textsc{povm}\xspace}

\newcommand\Lhs{\textsc{lhs}\xspace}
\newcommand\Rhs{\textsc{rhs}\xspace}
\newcommand\Locc{\textsc{locc}\xspace}
\newcommand\Mes{\textsc{mes}\xspace}
\newcommand\Lfm{\textsc{lfm}\xspace}

\newcommand\Jwt{\textsc{jwt}\xspace}

\newcommand\Car{\textsc{car}\xspace}

\newcommand\ie{i.e.\xspace}
\newcommand\eg{e.g.\xspace}

\renewcommand\qswap{S_\mathrm{Q}}
\newcommand\fswap{S_\mathrm{F}}
\newcommand\ppext[1]{\widetilde{#1}}
\allowdisplaybreaks

\begin{document}
\title{The Feynman problem and Fermionic entanglement:\\ Fermionic theory versus qubit theory}

\author{GIACOMO MAURO D'ARIANO}
\address{QUIT group, Dipartimento di Fisica, via Bassi 6,\\
Pavia, 27100, Italy.\\
INFN Gruppo IV, Sezione di Pavia, via Bassi 6,\\
Pavia, 27100, Italy.\\
dariano@unipv.it}

\author{FRANCO MANESSI}
\address{QUIT group, Dipartimento di Fisica, via Bassi 6,\\
Pavia, 27100, Italy.\\
franco.manessi01@ateneopv.it}

\author{PAOLO PERINOTTI}
\address{QUIT group, Dipartimento di Fisica, via Bassi 6,\\
Pavia, 27100, Italy.\\
INFN Gruppo IV, Sezione di Pavia, via Bassi 6,\\
Pavia, 27100, Italy.\\
paolo.perinotti@unipv.it}

\author{ALESSANDRO TOSINI}
\address{QUIT group, Dipartimento di Fisica, via Bassi 6,\\
Pavia, 27100, Italy.\\
alessandro.tosini@unipv.it}

\begin{abstract}
  The present paper is both a review on the Feynman problem, and an original research presentation
  on the relations between Fermionic theories and qubits theories, both regarded in the novel
  framework of operational probabilistic theories. The most relevant results about the Feynman
  problem of simulating Fermions with qubits are reviewed, and in the light of the new original
  results the problem is solved. The answer is twofold. On the computational side the two theories
  are equivalent, as shown by Bravyi and Kitaev \cite{Bravyi2002210}. On the operational side
  the quantum theory of qubits and the quantum theory of Fermions are different, mostly in the
  notion of locality, with striking consequences on entanglement. Thus the emulation does not respect
  locality, as it was suspected by Feynman in Ref.~\citen{feynman1982simulating}.
\end{abstract}

\maketitle
\bibliographystyle{ws-ijmpa}

\section{Introduction}

In the last three decades the relation between Fermionic systems and other quantum systems has been throughly
investigated from both the computational and the physical point of view. In particular the puzzling
anti-commuting nature of the Fermionic systems casts a shadow on the possibility of simulating the physical
evolution of a bunch of Fermionic systems by means of commuting quantum systems---say \emph{qubits}.  This
issue was raised by R.~P.~Feynman in 1982 \cite{feynman1982simulating}, when in his seminal work on physical
computation he wondered about the possibility of simulating Fermions by local quantum systems in
interaction---what we would call nowadays a {\em quantum computer} :
\begin{quote}
  \emph{
    Could we imitate every quantum mechanical system which is discrete and has a finite number of degrees of
    freedom? I know, almost certainly, that we could do that for any quantum mechanical system which involves
    Bose particles. I'm not sure whether Fermi particles could be described by such a system. So I leave that
    open.
  }
\end{quote}
The problem is that of encoding the evolution of Fermionic fields onto localized quantum systems.  A
well-known encoding of $N$ Fermionic systems into $N$ qubits is given by the \emph{Jordan-Wigner transform}
(\Jwt)\cite{jordan1928paulische}. Such an encoding, based on the identification between the Fock space of $N$
Fermions and the Hilbert space of $N$ qubits, provides a $*$-algebra isomorphism between the Fermionic
anticommuting algebra and the commuting algebra of qubits.  Such a correspondence has been a valuable
instrument in modern solid state physics for solving the one dimensional \textsc{xy} spin-chains
\cite{onsager1944crystal, lieb1964two} and then for the understanding of superconductivity and quantum Hall
effect. Moreover, a \emph{time-adaptive \Jwt} has been introduced in Ref.~\citen{PhysRevA.81.050303}, which
allows to contract Fermionic unitary circuits with the same complexity as for the corresponding spin model.
In quantum information science the \Jwt has been used to extend to the Fermionic case notions as entanglement
\cite{PhysRevA.76.022311}, entropic area law \cite{PhysRevLett.96.010404}, and universal computation
\cite{Bravyi2002210}. More recently the \Jwt, which originally regards one dimensional chains of
spin-$\tfrac{1}{2}$ systems, has been generalized to any spin \cite{PhysRevLett.86.1082} and lattice
\cite{PhysRevLett.71.3622} dimension.

Despite its computational power, the \Jwt fails to solve completely the issue established by Feynman:
physically local Fermionic operations are mapped into nonlocal quantum ones and vice versa.  As noticed by
many authors this can lead to ambiguities in defining the partial trace \cite{PhysRevA.83.062323,
PhysRevA.85.016301, PhysRevA.85.016302, PhysRevA.87.022338}, and in assessing the local nature of operations
\cite{verstraete2005mapping}.

Independently on the \Jwt the Fermionic systems are usually assumed to obey the \emph{Wigner superselection
rule}. Based on the simple argument of the impossibility of discriminating a $2\pi$ rotation from the identity
\cite{streater1964pct, weinberg1996quantum}, this superselection rule corresponds to an inhibition to the
superposition rule and  forbids superpositions among states with an odd number and an even number of Fermionic
excitations. Such a constraint on the admitted states for a set of Fermionic systems avoids the ambiguities
connected to the \Jwt \cite{PhysRevA.76.022311}, but it has never been shown to promote the Jordan-Wigner
isomorphism to a ``physical isomorphism''---\ie preserving some sort of locality of the Fermionic operations
through the encoding.

In this paper we tackle the issue of retaining locality of Fermionic operations through a qubit simulation in
a novel way, namely considering the Fermionic modes as the elementary systems of an \emph{operational
probabilistic theory} (\Opt).  The context of {\Opt}s provides a unified framework for studying and comparing
properties of different probabilistic models, such as locality.  Well-known examples of {\Opt}s are: (i)
quantum theory (\Qt) (recently axiomatized within the operational framework \cite{hardy2001quantum,
quit-purification, quit-derivation}), (ii) the classical information theory \cite{quit-derivation}, (iii) the
box-world \cite{PhysRevA.75.032304}, and (iv) the real quantum theory (\Rqt) \cite{stueckelberg1961quantum,
hardy2012limited}.  In Section \ref{sec:opt} we review the operational framework  and present the recent
results of Ref.  \citen{ d2013Fermionic}, where the superselection rule has been formalized in the general
context of {\Opt}s.

In Section \ref{s:fqt}  we build up the largest \Opt corresponding to the Fermionic computation.  We write all
possible events (states, transformations, effects) of the theory achieved with the anticommuting algebra of
the Fermionic field and assuming operations involving fields on some Fermionic modes to be local on those
modes. Locality here is meant in the operational sense, namely operations on systems that are not causally
connected must commute. The derivation leads naturally to the \emph{Wigner superselection rule}. Since there
is not a unique {\Opt} respecting such a superselection rule we then look for the largest theory compatible
with the locality of Fermionic operations, here denoted Fermionic quantum theory (\Fqt). 

In the second part of the paper (see Section \ref{sec:informational-features}) we study the operational
consequences of superselection. Unlike \Qt, \Fqt does not satisfy \emph{local tomography}, \ie the possibility
of  discriminating between two nonlocal states using only local measurements. After proving the correspondence
between Fermionic and qubit local operations with classical communication (\Locc), we study the emerging
notion of entanglement for Fermionic systems, an issue addressed in Ref.~\citen{PhysRevA.76.022311} for the
first time. Here we will identify non-separability as the unique notion of entanglement in \Fqt.  Upon
defining the Fermionic \emph{entanglement of formation} and \emph{concurrence}, we see that in \Fqt there are
states with maximal entanglement of formation that are mixed and that Fermionic entanglement does not satisfy
\emph{monogamy}, \ie the limitation on the sharing of entanglement between many parties.  Moreover the notion
of maximally entangled state must be replaced with the one of \emph{maximally entangled set}\cite{Kraus} also
in the bipartite case, unlike \Qt. Interestingly, while in \Qt  a simple linear criterion for full
separability of states is lacking we will see that \Fqt allows for it.

It is worth mentioning that \Fqt is only a special example of superselected \Qt while the notion of
superselection of Ref.~\citen{d2013Fermionic} allows for many other theories. Among them we will discuss
briefly the case of \Rqt---which also lacks local tomography\cite{hardy2012limited} and monogamy of
entanglement\cite{wootters2012monogamy}---and the theory with \emph{number superselection}---which only admits
superposition of states having the same particle occupation number.

A computational model based on Fermionic systems has already proposed by Bravyi and Kitaev in
Ref.~\citen{Bravyi2002210}. They showed that such a model supports \emph{universal computation} and that it is
equivalent to the qubit computational one. The computational model of Ref.~\citen{Bravyi2002210} is just the
\Fqt with the additional constraint given by the \emph{conservation of parity}; as a consequence the resulting
sets of transformations are strictly included in the \Fqt's ones. In Section~\ref{sec:fermionic-computation}
we compare \Qt and \Fqt from the point of view of computational complexity, and exploiting the results of
Bravyi and Kitaev\cite{Bravyi2002210} (here reviewed) we show the equivalence of the two theories and that
even \Fqt supports universal computation.

\section{Operational Probabilistic Theories}\label{sec:opt}

Before starting we need to review the basic definitions and notations for Operational Probabilistic Theories
(\Opt). For a detailed discussion see Ref.~\citen{quit-purification}. The fundamental notion in the
operational framework is that of \emph{test}, which is the abstract element of the framework corresponding to
a (single use) of a physical device. In more details, a test $\Test{A}=\{\Transformation{A}_{i}\}_{i\in\eta}$
describes an elementary operation that  usually produces an outcome $i$ belonging to the set $\eta$ of all the
possible outcomes. The readout of the outcome $i$ specifies the occurrence of the physical circumstance
identified by the \emph{event} $\Transformation{A}_{i}$.  Tests are also specified by an input and an output
label---\eg $\sA,\sB$---that identify the \emph{system types} (\emph{systems}, for short).  The test \(
\Test{A} \) and its building events \( \Transformation{A}_{i} \in \Test{A} \) can also be represented in the
following pictorial way:
\begin{align*}
  & \Test{A} \equiv 
    \Qcircuit @C=1em @R=.7em @! R { 
      & \ustickcool{\sA}\qw & \gate{\Test{A}} & \ustickcool{\sB}\qw & \qw
    }, 
  & & \Transformation{A}_i \equiv 
    \Qcircuit @C=1em @R=.7em @! R {
      & \ustickcool{\sA}\qw & \gate{\Transformation{A}_{i}} & \ustickcool{\sB}\qw &\qw
    }. 
\end{align*}
If an event $\Transformation{A}$ belongs to a singleton test \( \Test{A} \)---\ie \( \Test{A} = \{
\Transformation{A} \} \)---we say that \( \Transformation{A} \) is \emph{deterministic}.

Physical devices can be connected in sequence, as long as the output system type of each device is the same as
the input system type of the next one. So do tests: two tests $\Test{A} = \{ \Transformation{A}_i
\}_{i\in\eta}$, $\Test{B} = \{ \Transformation{B}_j \}_{j\in\chi}$ can be connected in sequence as long as the
output wire of the first one in the sequence (say \( \Test{A} \)) is of the same type as that of the input
wire of the last one (say \( \Test{B} \)), thus giving the \emph{sequential composition} \( \Test{B} \circ
\Test{A} \coloneqq \{ \Transformation{B}_j \circ \Transformation{A}_i \}_{(i,j)\in\eta\times\chi} \);
pictorially
\begin{align*}
  &
  \Qcircuit @C=1em @R=.7em @! R {
    & \ustickcool{\sA}\qw & \gate {\Test{A}} & \ustickcool{\sB}\qw & \gate{\Test{B}}&\ustickcool{\sC}\qw&\qw
  } 
  && 
  \Qcircuit @C=1em @R=.7em @! R {
    & \ustickcool{\sA}\qw & \gate {\Transformation{A}_{i}} & \ustickcool{\sB}\qw & 
      \gate{\Transformation{B}_{j}}&\ustickcool{\sC}\qw&\qw
  }
\end{align*}
The labels of the input and output systems provide rules for connecting tests in sequences. Notice that the
input/output relation has no causal connotation, and it does not entail an underlying ``time arrow''. As we
will see shortly, only in a causal \Opt it is possible to understand the input/output relation as a time
direction.

For every system $\sA$ there exists a unique singleton test $\{\Transformation{I}_{\sA} \}$ such that
$\Transformation{I}_{\sB} \circ\Transformation{A}=\Transformation{A}\circ\Transformation{I}_{\sA}$ for every
event $\Transformation{A}$ with input $\sA$ and output $\sB$. For every couple of systems $\sA$, $\sB$ we can
form the composite system $\sC\coloneqq\sA\sB$, on which we can perform tests $\Test A\otimes\Test B$ with
events $\Transformation{A}_i\otimes\Transformation{B}_j$ in \emph{parallel composition}, represented as
follows
\begin{equation*}
  \begin{aligned}
    \Qcircuit @C=1em @R=.7em @! R {
      & \qw \poloFantasmaCn \sA & \multigate{1} { \Transformation{A}_i \otimes \Transformation{B}_j } & \qw 
        \poloFantasmaCn \sB &\qw \\
      & \qw \poloFantasmaCn \sC & \ghost { \Transformation{A}_i \otimes \Transformation{B}_j } & \qw 
        \poloFantasmaCn \sD &\qw
    }
  \end{aligned}
  =
  \begin{aligned}
    \Qcircuit @C=1em @R=.7em @! R {
      & \qw \poloFantasmaCn \sA & \gate { \Transformation{A}_i } & \qw \poloFantasmaCn \sB & \qw \\
      & \qw \poloFantasmaCn \sC & \gate { \Transformation{B}_j } & \qw \poloFantasmaCn \sD &\qw
    }
  \end{aligned},
\end{equation*}
and satisfying the condition
\begin{equation*}
  \begin{aligned}
    \Qcircuit @C=1em @R=.7em @! R {
      & \qw \poloFantasmaCn \sA & \gate {\Transformation{A}_i } & \qw \poloFantasmaCn \sB & 
        \gate {\Transformation{B}_j } & \qw \poloFantasmaCn \sC & \qw \\
      & \qw \poloFantasmaCn \sD & \gate {\Transformation{C}_k } & \qw \poloFantasmaCn \sE & 
        \gate {\Transformation{D}_l } & \qw \poloFantasmaCn \sF & \qw
    }
  \end{aligned} 
  = 
  \begin{aligned}
    \Qcircuit @C=1em @R=.7em @! R {
      & \qw \poloFantasmaCn \sA & \gate {\Transformation{B}_j \circ \Transformation{A}_i } & 
        \qw \poloFantasmaCn \sC &\qw \\ 
      & \qw \poloFantasmaCn \sE & \gate {\Transformation{D}_l \circ \Transformation{C}_k} & 
        \qw \poloFantasmaCn \sF &\qw
    }
  \end{aligned},
\end{equation*}
in formulae \( ( \Transformation{B}_{j} \otimes \Transformation{D}_{l} ) \circ ( \Transformation{A}_{i}
\otimes \Transformation{C}_{k} ) = ( \Transformation{B}_j \circ \Transformation{A}_i ) \otimes (
\Transformation{D}_{l} \circ \Transformation{C}_{k} ) \).  Notice that the symbol $\otimes$ is just a formal
way to identify the parallel composition among tests (and events), and it is not the usual tensor product of
linear spaces. Moreover, the previous property implies commutation of tests on different systems, \ie for
every couple of events \( \Transformation{A}_i \), \( \Transformation{B}_j \) it is
\begin{equation}\label{eq:local-transformations-commute}
  \begin{aligned}
    \Qcircuit @C=1em @R=.7em @! R {
      & \qw \poloFantasmaCn \sA & \gate {\Transformation{A}_i} & \qw \poloFantasmaCn \sB &\qw\\ 
      & \qw \poloFantasmaCn \sC & \gate {\Transformation{B}_j } & \qw \poloFantasmaCn \sD &\qw
    }
  \end{aligned}=
  \begin{aligned}
    \Qcircuit @C=1em @R=.7em @! R {
      & \qw \poloFantasmaCn \sA & \gate {\Transformation{A}_i } & \qw \poloFantasmaCn \sB & 
        \gate {\Transformation{I}} & \qw \poloFantasmaCn \sB & \qw \\
      & \qw \poloFantasmaCn \sC & \gate {\Transformation{I} } & \qw \poloFantasmaCn \sC & 
        \gate {\Transformation{B}_j } & \qw \poloFantasmaCn \sD & \qw
    }
  \end{aligned}
  =
  \begin{aligned}
    \Qcircuit @C=1em @R=.7em @! R {
      & \qw \poloFantasmaCn \sA & \gate {\Transformation{I}} & \qw \poloFantasmaCn \sA & 
        \gate {\Transformation{A}_i} & \qw \poloFantasmaCn \sB & \qw \\
      & \qw \poloFantasmaCn \sC & \gate {\Transformation{B}_j } & \qw \poloFantasmaCn \sD & 
        \gate {\Transformation{I} } & \qw \poloFantasmaCn \sD & \qw
    }
  \end{aligned}.
\end{equation}

There is a special system type $\sI$, the \emph{trivial system}, such that $\sA\sI=\sI\sA=\sA$.  The tests
with input system $\sI$ and output $\sA$ are called {\em preparation-tests} of $\sA$, while the tests with
input system $\sA$ and output $\sI$ are called {\em observation-tests} of $\sA$.  Preparation-events of $\sA$
are denoted by the symbols \( \State{\rho}_\sA \) or \( \Qcircuit @C=.5em @R=.5em { \prepareC{\rho} &
\ustickcool{\sA} \qw & \qw } \), and observation-events by \( \Effect{ c}_\sA \) or \( \Qcircuit @C=.5em
@R=.5em { & \ustickcool{\sA} \qw & \measureD{ c} } \).

An arbitrary complex test obtained by parallel and/or sequential composition of ``elementary tests'' is called
\emph{circuit}. An \emph{operational theory} is a collection of systems closed under composition, and a
collection of tests closed under parallel and sequential composition, \ie every circuit belongs to the theory.
Given a circuit we say that an event $\Transformation{H}$ {\em is immediately connected to the input of}
$\Transformation{K}$, and write $\Transformation{H}\prec_1\Transformation{K}$, if there is an output system of
$\Transformation{H}$ that is connected with an input system of $\Transformation{K}$; \eg in
Fig.~\ref{fig:closed-circuit}, \( \Transformation{A}_{i_2} \prec_1 \Transformation{D}_{i_5} \). Moreover we
can introduce the transitive closure $\prec$ of the relation $\prec_1$, and we say that $\Transformation{H}$
{\em is connected to the input of} $\Transformation{K}$ if $\Transformation{H}\prec\Transformation{K}$ (\eg \(
\Transformation{B}_{i_3} \prec \Transformation{E}_{i_6} \)). The two relations \( \prec_1 \) and \( \prec \)
can be extended to tests trivially.

\begin{figure}
  \begin{equation*}
    \begin{aligned}
      \Qcircuit @C=1em @R=.7em @! R {
        \multiprepareC{3}{\Psi_{i_1}} & \qw \poloFantasmaCn{\sA} & \multigate{1}{\Transformation{A}_{i_2}} &
          \qw \poloFantasmaCn{\sB} & \gate{\Transformation{C}_{i_4}} & \qw \poloFantasmaCn{\sC} & 
          \multigate{1}{\Transformation{E}_{i_6}} & \qw \poloFantasmaCn{\sD} & \multimeasureD{2}{G_{i_8}} \\
        \pureghost{\Psi_{i_1}} & \qw \poloFantasmaCn{\System{E}} & \ghost{\Transformation{A}_{i_2}} & 
          \qw \poloFantasmaCn{\System{F}} & \multigate{1}{\Transformation{D}_{i_5}} & 
          \qw \poloFantasmaCn{\System{G}} & \ghost{\Transformation{E}_{i_6}} & & \pureghost{G_{i_8}} \\
        \pureghost{\Psi_{i_1}} & \qw \poloFantasmaCn{\System{H}} & \multigate{1}{\Transformation{B}_{i_3}} &
          \qw \poloFantasmaCn{\System{L}} & \ghost{\Transformation{D}_{i_5}} & \qw \poloFantasmaCn{\System{M}}
          & \multigate{1}{\Transformation{F}_{i_7}} & \qw \poloFantasmaCn{\System N} & \ghost{G_{i_8}} \\
        \pureghost{\Psi_{i_1}} & \qw \poloFantasmaCn{\System{O}} & \ghost{\Transformation{B}_{i_3}} & 
          \qw \poloFantasmaCn{\System{P}} & \qw & \qw & \ghost{\Transformation{F}_{i_7}} \\
      }
    \end{aligned}
  \end{equation*}
  \caption{ The closed circuit in the figure represents the joint probability \( \ProbCond{i_1,i_2,\dots
    i_8}{\PreparationTest{\Psi},\Test A,\dots,\ObservationTest{G}} \) of outcomes $i_1,i_2,\dots i_8$ conditioned
    by the choice of tests ${\PreparationTest{\Psi},\Test A,\dots,\ObservationTest{G}}$. Since the output of the
    event \( \Transformation{A}_{i_2} \) is connected to the input of the event \( \Transformation{D}_{i_5} \)
    through the system \( \System F \), the event \( \Transformation{A}_{i_2} \) immediately precedes the event \(
    \Transformation{D}_{i_5} \) (\( \Transformation{A}_{i_2} \prec_1 \Transformation{D}_{i_5} \)).  Similarly,
    since between the event \( \Transformation{B}_{i_3} \) and the event \( \Transformation{E}_{i_6} \) there is
    \( \Transformation{D}_{i_5} \) such that \( \Transformation{B}_{i_3} \prec_1 \Transformation{D}_{i_5}\prec_1
    \Transformation{E}_{i_6}  \), the event \( \Transformation{B}_{i_3} \) precedes the event \(
    \Transformation{E}_{i_6} \) (\( \Transformation{B}_{i_3} \prec \Transformation{E}_{i_6} \)). If the closed
    circuit of the figure belongs to a causal theory, we have that the marginal probability of the event \(
    \Transformation{D}_{i_5}\in\Test{D} \) cannot depend on the choice of any test \( \Test{X} \) such that \(
    \Test{X} \not\prec \Test{D} \), \ie \( \ProbCond{ i_5 }{ \PreparationTest{\Psi}, \Test{A}, \Test{B}, \Test{C},
    \Test{D}, \Test{E}, \Test{F}, \ObservationTest{G} } = \ProbCond{ i_5 }{ \PreparationTest{\Psi},\Test{A},
    \Test{B} } \).
    \label{fig:closed-circuit}
  }
\end{figure}
A circuit is \emph{closed} if its overall input and output systems are the trivial ones. Figure
\ref{fig:closed-circuit} is an example of a closed circuit.  An {\em operational probabilistic theory} is an
operational theory where every closed circuit represents a probability distribution; \eg the closed circuit in
Fig.~\ref{fig:closed-circuit} represents the probability \( \ProbCond{i_1,i_2,\dots i_8}{\bm\Psi,\Test
A,\dots,\bm G} \) of outcomes $i_1,i_2,\dots i_8$ conditioned by the choice of tests ${\bm\Psi,\Test
A,\dots,\bm G}$. In probabilistic theories we can quotient the set of preparation-events of $\sA$ by the
equivalence relation \( \State{\rho}_\sA\sim\State{\sigma}_\sA \Leftrightarrow \) \emph{the probability of
preparing} \( \State{\rho}_\sA \) \emph{and measuring} \( \Effect{c}_\sA \) \emph{is the same as that of
preparing} \( \State{\sigma}_\sA \) \emph{and measuring} \( \Effect{c}_\sA \) \emph{for every
observation-event} \( \Effect{c}_\sA \) \emph{of} \( \sA \) (and similarly for observation-events).  The
equivalence classes of preparation-events and observation-events of \( \sA\) will be denoted by the same
symbols as their elements $\State{\rho}_\sA$ and $\Effect{c}_\sA$, respectively, and will be called
\emph{state} $\State{\rho}_\sA$ for system $\sA$, and \emph{effect} $\Effect{c}_\sA$ for system $\sA$.  For
every system $\sA$, we will denote by $\SetStates(\sA)$, $\SetEffects(\sA)$ the sets of states and effects,
respectively. States are real-valued functionals over the effects, and vice versa; thus they can be embedded
respectively in the real vector spaces \( \SetStates_{\Reals}(\sA) \), \( \SetEffects_{\Reals}(\sA) \). \(
\SetStates_{\Reals}(\sA) \) is the dual space of \( \SetEffects_{\Reals}(\sA) \), and vice versa since the
dimension $D_\sA \coloneqq \dim \SetEffects_{\Reals}(\sA)$ is assumed to be finite.  The application of the
effect $\Effect{c_i}_\sA$ on the state $\State{\rho}_\sA$ is written as $\RBraKet{ c_i }{ \rho }_\sA$ and
corresponds to the closed circuit $\Qcircuit @C=.5em @R=.5em { \prepareC{\rho} & \ustickcool{\sA} \qw &
\measureD{ c_i } } $, denoting therefore the probability of the $i$th outcome of the observation-test $\bm
c=\{\Effect{c_i}_\sA\}_{i\in\eta}$ performed on the state $\rho$ of system $\sA$, \ie $ \RBraKet{ c_i }{ \rho
}_\sA \coloneqq\ProbCond{c_i}{\bm\rho, \bm c}$.

Any event with input system $\sA$ and output system $\sB$ induces a collection of linear mappings from \(
\SetStates_{\mathbb R}(\sA\sC) \) to \( \SetStates_{\mathbb R}(\sB\sC) \), for varying system $\sC$. Such a
collection is called {\em transformation} from $\sA$ to $\sB$. The set of transformations from $\sA$ to $\sB$
will be denoted by $\SetTransf(\sA,\sB)$, and its linear span by $\SetTransf_\Reals(\sA,\sB)$. The symbols \(
\Transformation{A} \) and $ \Qcircuit @C=.5em @R=.5em { & \ustickcool{\sA} \qw & \gate{\Transformation{A}} &
\ustickcool{\sB} \qw &\qw} $ denoting the event $\Transformation{A}$ will be also used to represent the
corresponding transformation.

One usually requires that an experimenter can randomize the choice of the devices in an experiment with
arbitrary probabilities. This implies that, for every system, all the set of states, effects, and
transformations of an \Opt are convex. The extremal points of the convex set of the deterministic states (and
similarly for effects and transformations) correspond to the so-called \emph{atomic states}, also known as
\emph{pure states} since they cannot be seen as convex combinations of other deterministic states.

An \Opt can satisfy many different properties\cite{quit-derivation}; among the most important ones there is
the property of \emph{causality}.
\begin{definition}\label{def:causality}
  An \Opt is \emph{causal} if for every preparation-test $\PreparationTest{\rho} =\{\State{\rho_i}\}_{i \in
  \eta}$ and any two observation-tests $\ObservationTest{a}=\{\Effect{a_j}\}_{j \in \chi}$ and
  $\ObservationTest{b}=\{\Effect{b_j}\}_{j \in \xi}$ one has \( \sum_{j \in \chi} \RBraKet{a_j}{\rho_i} =
  \sum_{k \in \xi} \RBraKet{b_k} {\rho_i},\forall i\in\eta\), namely the probability of the preparation is
  independent of the choice of observation.
\end{definition}
Causality is equivalent to the so-called \emph{no-backward signaling}
\cite{PhilosophyQuantumInformationEntanglement2}, namely within a closed circuit, the marginal probabilities
of the outcomes of an arbitrary test $\Test{H}$ do not depend on the choice of any test
$\Test{K}\not\prec\Test{H}$.  For example, in the circuit of Fig.~\ref{fig:closed-circuit} causality implies
\begin{equation*}
  \ProbCond{ i_5 }{ \PreparationTest{\Psi}, \Test{A}, \Test{B}, \Test{C}, \Test{D}, \Test{E}, \Test{F},
  \ObservationTest{G} } = \ProbCond{ i_5 }{ \PreparationTest{\Psi}, \Test{A}, \Test{B} }.
\end{equation*}
The present notion of causality is a rigorous definition in the operational framework of the so-called {\em
Einstein causality}.  Indeed, a corollary of \emph{no-backward signaling} is the {\em no-signaling without
interaction} \cite{quit-purification}. In an \Opt, a condition equivalent to causality is that of uniqueness
of the deterministic effect \cite{quit-purification} (usually denoted by $\Effect{e}$). Notice that given a
bipartite state $\State{\rho}_{\sA\sB}$, the deterministic effects $\Effect{e}_\sB$ and $\Effect{e}_\sA$ allow
one to evaluate the marginal states (\eg partial trace in \Qt) on the component systems $\sA$ and $\sB$
\begin{equation*}
  \begin{aligned}
    \Qcircuit @C=1em @R=.7em @! R {\prepareC{\rho}& \qw \poloFantasmaCn \sA &\qw }
  \end{aligned}
  ~=~
  \begin{aligned}
    \Qcircuit @C=1em @R=.7em @! R {
      \multiprepareC{1}{\rho}& \qw \poloFantasmaCn \sA &\qw \\
      \pureghost{\rho} & \qw \poloFantasmaCn \sB &\measureD {e_\sB}
    } 
  \end{aligned}
  \,, \qquad
  \begin{aligned}
    \Qcircuit @C=1em @R=.7em @! R {\prepareC{\sigma}& \qw \poloFantasmaCn \sB &\qw }
  \end{aligned}
  ~=~
  \begin{aligned}
    \Qcircuit @C=1em @R=.7em @! R {
      \multiprepareC{1}{\rho}& \qw \poloFantasmaCn \sA &\measureD {e_\sA} \\
      \pureghost{\rho} & \qw \poloFantasmaCn \sB &\qw
    } 
  \end{aligned}
  \,.
\end{equation*}

Another property is the so-called \emph{no-restriction hypothesis}. We say that a linear map \(
\Transformation{A} \in \SetTransf_{\Reals}(\sA,\sB) \) is {\em admissible} if it locally preserves the set of
states \( \SetStates(\sA\sC) \) for every ancillary system $\sC$; namely 
\begin{equation*}
  \begin{aligned}
    \Qcircuit @C=1em @R=.7em @! R { 
      \multiprepareC{1}{\rho} & \ustickcool{\sA} \qw & \qw \\
      \pureghost{\rho} & \ustickcool{\sC} \qw & \qw 
    }
  \end{aligned} \in \SetStates(\sA\sC)
  \implies 
  \begin{aligned}
    \Qcircuit @C=1em @R=.7em @! R { 
      \multiprepareC{1}{\rho} & \ustickcool{\sA} \qw & \gate{\Transformation{A}} & \ustickcool{\sB} \qw & \qw \\
      \pureghost{\rho} & \ustickcool{\sC} \qw & \gate{\Transformation{I}} & \ustickcool{\sC} \qw & \qw
    }
  \end{aligned} \in \SetStates(\sB\sC) \qquad \forall \sC.
\end{equation*}
The no-restriction hypothesis requires that every admissible map in \( \SetTransf_\Reals(\sA,\sB) \) actually
belongs to \( \SetTransf(\sA,\sB) \)
\footnote{In previous literature \cite{quit-purification} the same nomenclature has been used for a different
  concept: for every system $\sA$ the convex cone generated by \( \SetEffects(\sA)
  \) coincides with the dual convex cone generated by the set of states \( \SetStates(\sA) \).}. Notice that an
\Opt satisfying the no-restriction hypothesis is completely determined by its systems and the respective set
of states, since even the effects---being particular kind of transformations---are all the admissible ones.
We can therefore say that a no-restricted \Opt is simply the collection \( \Theta \coloneqq \{ ( \System{X},
\SetStates(\System{X}) ) \}_{\System{X} \in \eta} \) for varying system \( \System{X} \).

\subsection{Local, bilocal, \ldots, $n$-local tomography}\label{sub:lobi}

A common assumption in the literature of probabilistic theories is the so-called \emph{local tomography} (also
called by some authors \emph{local discriminability} or \emph{local distinguishability}); namely the
possibility of distinguishing two different bipartite states, by means of local devices.
\begin{definition}
  A theory enjoys local tomography if for any \( \State{\rho}, \State{\sigma} \in \SetStates(\sA\sB) \) we
  have
  \begin{equation*}
    \State{\rho} \ne \State{\sigma} \implies \exists \Effect{a} \in \SetEffects(\sA),
    \Effect{b} \in \SetEffects(\sB) \text{ such that }
    \begin{aligned}
      \Qcircuit @C=1em @R=.7em @! R { 
        \multiprepareC{1}{\rho} & \ustickcool{\sA} \qw & \measureD{a} \\
        \pureghost{\rho} & \ustickcool{\sB} \qw & \measureD{b}
      }
    \end{aligned}
    \ne
    \begin{aligned}
      \Qcircuit @C=1em @R=.7em @! R { 
        \multiprepareC{1}{\sigma} & \ustickcool{\sA} \qw & \measureD{a} \\
        \pureghost{\sigma} & \ustickcool{\sB} \qw & \measureD{b}
      }
    \end{aligned}.
  \end{equation*}
\end{definition}
An \Opt with local tomography allows to perform tomography on multipartite states with only local
measurements. Indeed, in such a scenario every bipartite effect \( \Effect{c}_{\sA\sB} \) can be written as
linear combination of product effects, therefore every probability \( \RBraKet{c}{\rho}_{\sA\sB} \) can be
computed as a linear combination of the probabilities \( ( \Effect{a}_\sA \otimes \Effect{b}_\sB )
\State{\rho}_{\sA\sB} \) arising from a finite set of product effects. In other words, we have the property
that the linear space of effects  of a composite system is actually the tensor product of the linear spaces of
effects of the component systems, \ie \( \SetEffects_{\Reals}(\sA\sB) \equiv \SetEffects_{\Reals}(\sA) \otimes
\SetEffects_{\Reals}(\sB) \). Since \( \SetStates_\Reals(\sA\sB) = \SetEffects_\Reals(\sA\sB)^* \) we have
that the same result holds also for the linear space of states. Thus, in a local-tomographic \Opt the parallel
composition of two states (effects) denoted by the symbol \( \otimes \) can be in fact understood as a tensor
product, moreover the following relation between the dimension of the set of states/effects holds: \(
D_{\sA\sB}=D_{\sA}D_{\sB} \).  

\begin{remark}\label{rem:local-tomography}
An important consequence of local tomography is that a transformation $ \Transformation{T} \in \SetTransf(
\sA,\sB )$ is completely specified by its action on $\SetStates(\sA)$:\cite{quit-purification}:
\begin{equation*}
  \Transformation{C} \State{\rho} = \Transformation{C^\prime} \State{\rho} \quad \forall \State{\rho}
  \in\SetStates(\sA) \ \Rightarrow\ \Transformation{C} = \Transformation{C^\prime}.
\end{equation*}
\end{remark}

One can imagine to relax the property of local tomography in many different ways; the most general scenario is
given by the \emph{$n$-local tomography}\cite{hardy2012limited}. First, we define an effect to be
\emph{$n$-local} if it can be written as a conic combination of tensor products of effects that are at most
$n$-partite.
 
\begin{definition}\label{def:n-local-tomography}
  A theory enjoys $n$-local tomography if whenever two states \( \State{\rho} \), \( \State{\sigma} \) are
  different, there is a $n$-local effect \( \Effect{a} \) such that \( \RBraKet{a}{\rho} \ne \RBraKet{a}{\sigma}
  \).
\end{definition}
Clearly, local tomography is the particular case of $n$-local tomography with $n=1$. Given a
$n$-local-tomographic theory with $n>1$, for an arbitrary bipartite system \( \sA\sB \) one has \( D_{\sA\sB}
\geq D_{\sA}D_{\sB} \), since in general \( \SetStates_\Reals(\sA\sB) = \SetStates_\Reals(\sA) \otimes
\SetStates_\Reals(\sB) ~\oplus~\SetStates_\Reals^{\textrm{NL}}(\sA\sB) \), where \(
\SetStates_\Reals^{\textrm{NL}}(\sA\sB) \coloneqq ( \SetStates_\Reals(\sA) \otimes \SetStates_\Reals(\sB)
)^\perp \) is the subspace where the non-local components of the bipartite states live.  By definition, a
$n$-local-tomographic theory is also $(n+1)$-local-tomographic, since a $n$-local effect is also
$(n+1)$-local. We are interested in {\Opt}s that are {\em strictly $n$-local-tomographic}, namely
$n$-local-tomographic {\Opt}s that are not $(n-1)$-local-tomographic.

Another case already studied in literature is \emph{bilocal tomography}\cite{hardy2012limited}, namely
$2$-local tomography. In particular, for such a case we have that for every couple of different tripartite
states \( \State{\rho}, \State{\sigma} \in \SetStates(\sA\sB\sC) \) there exist a $2$-local effect \(
\Effect{x} \in \SetEffects(\sA\sB\sC) \) such that
\begin{equation*}
  \begin{aligned}
    \Qcircuit @C=1em @R=.7em @! R { 
      \multiprepareC{2}{\rho} & \ustickcool{\sA} \qw & \multimeasureD{2}{x} \\
      \pureghost{\rho} & \ustickcool{\sB} \qw & \pureghost{x} \qw \\
      \pureghost{\rho} & \ustickcool{\sC} \qw & \pureghost{x} \qw 
    }
  \end{aligned}
  \ne
  \begin{aligned}
    \Qcircuit @C=1em @R=.7em @! R { 
      \multiprepareC{2}{\sigma} & \ustickcool{\sA} \qw & \multimeasureD{2}{x} \\
      \pureghost{\sigma} & \ustickcool{\sB} \qw & \pureghost{x} \qw \\
      \pureghost{\sigma} & \ustickcool{\sC} \qw & \pureghost{x} \qw 
    }.
  \end{aligned}
\end{equation*}
Notice that, since \( \Effect{x} \) is $2$-local, it can be written as the following conic combination
\begin{equation*}
  \begin{aligned}
    \Qcircuit @C=1em @R=.7em @! R { 
      & \ustickcool{\sA} \qw & \multimeasureD{2}{x} \\
      & \ustickcool{\sB} \qw & \pureghost{a} \qw \\
      & \ustickcool{\sC} \qw & \pureghost{a} \qw 
    }
  \end{aligned}
  = 
  \sum_{j} q_j
  \begin{aligned}
    \Qcircuit @C=1em @R=.7em @! R { 
      & \ustickcool{\sA} \qw & \measureD{a_j} \\
      & \ustickcool{\sB} \qw & \measureD{b_j} \\
      & \ustickcool{\sC} \qw & \measureD{c_j}
    }
  \end{aligned}
  +
  q_j^\prime
  \begin{aligned}
    \Qcircuit @C=1em @R=.7em @! R { 
      & \ustickcool{\sA} \qw & \multimeasureD{1}{d_j} \\
      & \ustickcool{\sB} \qw & \pureghost{d_j} \qw \\
      & \ustickcool{\sC} \qw & \measureD{e_j}
    }
  \end{aligned} 
  +
  q_j^{\prime\prime}
  \begin{aligned}
    \Qcircuit @C=1em @R=.7em @! R { 
      & \ustickcool{\sA} \qw & \measureD{f_j} \\
      & \ustickcool{\sB} \qw & \multimeasureD{1}{g_j} \\
      & \ustickcool{\sC} \qw & \pureghost{g_j} \qw
    }
  \end{aligned} 
  +
  q_j^{\prime\prime\prime}
  \begin{aligned}
    \Qcircuit @C=1em @R=.7em @! R { 
      & \ustickcool{\sA} \qw & \qw & \multimeasureD{2}{h_j} \\
      & \ustickcool{\sC} \qw & \measureD{i_j} \\
      & \ustickcool{\sB} \qw & \qw & \pureghost{h_j} \qw
    }
  \end{aligned} 
\end{equation*}
with \( q_j, q^\prime_j, q^{\prime\prime}_j, q^{\prime\prime\prime}_j \geq 0 \) and \( \Effect{d_j} \in
\SetEffects_\Reals^{\textrm{NL}}(\sA\sB) \) (and similarly for \( \Effect{g_j} \), \( \Effect{h_j} \)). For a
bilocal-tomographic theory we have therefore 
\begin{align}\label{eq:bilocal-lower}
  & D_{\sA\sB} \geq D_\sA D_\sB, \\
  & D_{\sA\sB\sC} \leq D_\sA D_\sB D_\sC + D_\sA \tilde D_{\sB\sC} +  D_\sB \tilde D_{\sA\sC} +D_\sC
    \tilde D_{\sB\sC} \label{eq:bilocal-upper}
\end{align}
where \( \dim (\SetEffects_\Reals^{\textrm{NL}}(\sA\sB)) =: \tilde D_{\sA\sB} = D_{\sA\sB} - D_{\sA}D_{\sB}
\).  A strictly bilocal-tomographic theory has the first bound tight, moreover if the upper bound is saturated
we say that the \Opt is {\em maximally bilocal-tomographic}, since it requires all the $2$-local effects to
discriminate multipartite states.

\subsection{Superselected operational probabilistic theories}\label{s:superselected-opt}

A superselection rule $\sigma$ on a theory $\Theta$ corresponds to a linear section of all sets of
transformations for each multipartite system, which under the no-restriction hypothesis reduces to sectioning
linearly just the sets of states.  We can give the following formal definition of superselection rule:
\begin{definition}\label{def:superselection-rule}
  A superselection rule $\sigma$ is a map from an \Opt $\Theta$ to another \Opt $\bar{\Theta}$,
    \begin{align*}
      \sigma:\ & \Theta\rightarrow\bar{\Theta},\\
               & (\sA, \SetStates(\sA) ) \mapsto \sigma(\ ( \sA, \SetStates(\sA) )\ ) =: 
                  (\bar\sA, \SetStates(\bar\sA)),
    \end{align*}
  such that, for every system \( \sA \), \( \SetStates(\bar\sA) \) is a linear section of \( \SetStates(\sA)
  \), \ie
  \begin{equation*}\label{eq:constraints}
    \SetStates(\bar{\sA}) \coloneqq \{ \rho\in\SetStates(\sA) \mid \RBraKet{s_i^\sigma}{\rho}= 0,\quad i=1,\ldots,
    V^\sigma_\sA \},
  \end{equation*}
  where \(  \Effect{s_i^\sigma}\in\SetEffects_{\Reals}(\sA) \) are $V^\sigma_\sA $  linear independent
  constraints.
\end{definition}

For consistency, the superselection map $\sigma$ must commute with system composition, forcing the definition
of composition for the constrained theory as $\sigma(\sA)\sigma(\sB) \coloneqq \sigma(\sA\sB)$.  Notice that,
being linear $\sigma$ preserves convexity of the theory, \ie all the sets \( \SetStates(\bar\sA) \), \(
\SetEffects(\bar\sA)  \), \( \SetTransf(\bar\sA,\bar\sB)  \), for every system \( \bar\sA \), \( \bar\sB \) of
the constrained theory are convex. For instance, this means that in a \Qt with superselection, states from
different sectors cannot be superimposed, but can be mixed.  From the definition, it follows immediately  \(
\SetStates(\bar\sA) \subseteq \SetStates(\sA) \), \( \SetEffects(\bar\sA)\subseteq\SetEffects(\sA) \), and \(
D_{\bar\sA} = D_{\sA} - V_{\sA}^\sigma \).

The number \( V_\sA^\sigma \) of linearly independent constraints on a system \( \sA \) cannot be arbitrary,
for example consider the trivial bound \( V_\sA^\sigma \leq D_\sA \). In fact, one has other more interesting
bounds due to the system composition.
\begin{proposition}\label{prop:constraints-consistency}
  Let $\bar\Theta$ be the superselected \Opt build from the \Opt \( \Theta \) by means of the superselection
  map $\sigma$. Then the following bounds hold:
  \begin{align}
    \label{eq:constraints-on-the-constraints-a}
      & V_{\sA\sB}^{\sigma} \geq D_{\sA} V_{\sB}^{\sigma} + D_{\sB} V_{\sA}^{\sigma} - 2 V_{\sA}^{\sigma} 
        V_{\sB}^{\sigma}, \\ 
    \label{eq:constraints-on-the-constraints-b}
      & V_{\sA\sB}^{\sigma} \leq D_{\sA} V_{\sB}^{\sigma} + D_{\sB} V_{\sA}^{\sigma} - V_{\sA}^{\sigma} 
        V_{\sB}^{\sigma} + D_{\sA\sB} - D_{\sA} D_{\sB}.
  \end{align}
\end{proposition}
\begin{proof}
  The upper bound of Eq.~\eqref{eq:constraints-on-the-constraints-b} is easily proven upon noticing that for an
  arbitrary \Opt it always happens that \( \SetStates_\Reals( \sA\sB ) \supseteq \SetStates_\Reals(\sA) \otimes
  \SetStates_\Reals(\sB) \), and thus \( D_{\sA\sB} \geq D_{\sA} D_{\sB} \). Hence, one has \(
  D_{\bar\sA\bar\sB} \geq D_{\bar\sA} D_{\bar\sB} \), and using \( D_{\bar\sB} = D_\sB - V_\sB^\sigma \) and \(
  D_{\bar\sB} = D_\sB - V_\sB^\sigma \) we get Eq.~\eqref{eq:constraints-on-the-constraints-b}.

  The lower bound of Eq.~\eqref{eq:constraints-on-the-constraints-a} is proved by showing that all the local
  constraints on the component systems $\sA$ and $\sB$ are also constraints of the composite system $\sA\sB$,
  namely for any $\Effect{b}\in\SetEffects(\bar\sB)$ and any $i=1,\ldots, V^\sigma_\sA$, one has that $
  \Effect{s_i^\sigma} \otimes \Effect{b} \in \SetEffects_{\Reals}(\sA\sB) $ is a constraint for
  $\bar\sA\bar\sB$.  Indeed, suppose by contradiction that \( \RBraKet{s_i^\sigma \otimes b }{\rho} \ne 0 \)
  for some \( \rho \in \SetStates(\bar\sA\bar\sB) \), and \( i \in \{ 1, \ldots, V_{\sA\sB}^\sigma  \} \).
  Since 
  \begin{equation*}
    \State{\rho_b} \coloneqq
    \begin{aligned}
      \Qcircuit @C=1em @R=.7em @! R {
        \multiprepareC{1}{\rho} & \ustickcool{\sA} \qw & \qw \\
        \pureghost{\rho} & \ustickcool{\sB} \qw & \measureD{b}
      }
    \end{aligned} 
  \end{equation*}
  is a valid state for the system \( \bar\sA \), we have $\RBraKet{s_i^\sigma}{\rho_b} \neq 0$ against the
  hypothesis. The same argument holds reversing the roles of the subsystems $\bar\sA$ and $\bar\sB$, so we
  conclude that $ V_{\sA\sB}^\sigma $ shall be at least \(  D_{\bar\sA} V_{\sB}^{\sigma} +
  D_{\bar\sB}V_{\sA}^{\sigma} \), which gives the lower bound of
  Eq.~\eqref{eq:constraints-on-the-constraints-a} using \( D_{\bar\sB} = D_\sB - V_\sB^\sigma \) and \(
  D_{\bar\sB} = D_\sB - V_\sB^\sigma \).
\end{proof} 

Given an \Opt $\Theta$ one can build ``bottom-up'' a superselected theory \( \bar\Theta \) by defining the
constraints only for the elementary systems (the ones that cannot be obtained by composition of other systems)
and taking the minimal number of linear constraints \eqref{eq:constraints-on-the-constraints-a} on the
composite ones. We call such superselected {\Opt}s {\em minimally superselected}.
\begin{definition}\label{def:minimal-superselection}
  A superselected \Opt is minimally superselected if it saturates the lower bound of
  Eq.~\eqref{eq:constraints-on-the-constraints-a}. 
\end{definition}
In a minimally superselected \Opt the only constraints on bipartite systems are $\Effect{r_i^\sigma} \otimes
\Effect{b}$ and $\Effect{a} \otimes \Effect{s_j^\sigma}$, with $a\in\SetEffects(\bar\sA)$,
$b\in\SetEffects(\bar\sB)$, $r_i^\sigma\in\SetEffects_\Reals(\sA)$, $r_j^\sigma\in\SetEffects_\Reals(\sB)$.

On the other hand, the saturation of the upper bound of Eq.~\eqref{eq:constraints-on-the-constraints-b} leads
to a {\em maximally superselected} \Opt:
\begin{definition}\label{def:maximal-superselection}
  A superselected \Opt is maximally superselected if it saturates the upper bound of
  Eq.~\eqref{eq:constraints-on-the-constraints-b}.
\end{definition}

Since enforcing superselection constraints on a \Opt leads to a change of the structure of the set of states,
effects, and transformations, we shall expect a change also in the properties satisfied by the resulting
theory. Indeed, while a causal theory retains causality once superselected, the converse is not true.
Moreover, a local tomographic theory is in general no more local tomographic upon superseletion, as the
following proposition shows.

\begin{proposition}\label{prop:minimal-maximal}
  Let $\bar\Theta$ be a superselection of a local-tomographic theory $\Theta$. Then:
  \begin{enumerate}[(i)]
    \setlength{\itemsep}{1pt}
    \setlength{\parskip}{2pt}
    \setlength{\parsep}{1pt}
    \item Minimal superselection $\Rightarrow$ $\bar\Theta$ maximally bilocal-tomographic,
    \item Maximal superselection $\Rightarrow$ $\bar\Theta$ local-tomographic.
  \end{enumerate}
\end{proposition}
\begin{proof}
  Let us prove the first implication. The superselected theory $\bar\Theta$ is maximally bilocal-tomographic if
  it saturates the bound of Eq.~\eqref{eq:bilocal-upper}, namely
  \begin{equation*}
    D_{\bar\sA\bar\sB\bar\sC} = D_{\bar\sA} D_{\bar\sB} D_{\bar\sC} + D_{\bar\sA} \tilde D_{\bar\sB\bar\sC}
    +  D_{\bar\sB} \tilde D_{\bar\sA\bar\sC} +D_{\bar\sC} \tilde D_{\bar\sB\bar\sC}.
  \end{equation*}
  We prove this equality evaluating the \Lhs\ and the \Rhs\ of the equation and enforcing the minimal
  superselection given by the lower bound of Eq.~\eqref{eq:constraints-on-the-constraints-a}
  \begin{equation}\label{eq:minimal-superselection}
  V_{\sA\sB}^{\sigma} = D_{\sA}V_{\sB}^{\sigma} + D_{\sB}V_{\sA}^{\sigma} - 2 V_{\sA}^{\sigma}V_{\sB}^{\sigma}.
  \end{equation}
  
  \Lhs: we have \( D_{\overline{\sA\sB\sC}}=D_{\sA\sB\sC}-V^{\sigma}_{\sA\sB\sC} \); taking the partition
  $\bar\sA\bar\sB\bar\sC=\bar\sA(\bar\sB\bar\sC)$, the requirement of minimal
  superselection
  gives
  \begin{equation*}
    D_{\bar\sA\bar\sB\bar\sC} = D_{\sA\sB\sC} - ( D_{\sA} V^{\sigma}_{\sB\sC} + D_{\sB\sC} V^{\sigma}_{\sA} - 2
    V^{\sigma}_{\sA} V^{\sigma}_{\sB\sC}).
  \end{equation*}
  Using again the minimal superselection requirement, we expand 
  $V^{\sigma}_{\sB\sC}$ and $V^{\sigma}_{\sB\sC}$ getting
  \begin{multline*}
   D_{\bar\sA\bar\sB\bar\sC} = D_{\sA} D_{\sB} D_{\sC}  - ( D_{\sA} D_{\sB} V^{\sigma}_{\sC} + D_{\sA} D_{\sC}
      V^{\sigma}_{\sB} + D_{\sB} D_{\sC} V^{\sigma}_{\sA}) \\
    + 2 ( D_{\sA} V^{\sigma}_{\sB} V^{\sigma}_{\sC} + D_{\sA} V^{\sigma}_{\sC}
    V^{\sigma}_{\sB} + D_{\sB} V^{\sigma}_{\sC} V^{\sigma}_{\sA} ) 
  - 4 V^{\sigma}_{\sA} V^{\sigma}_{\sB} V^{\sigma}_{\sC},
  \end{multline*}
  where we used the identity $D_{\sA\sB}=D_{\sA}D_{\sB}$ since the \Opt \( \Theta \) is local tomographic.
  
  \Rhs: We use the identities \( \tilde D_{\bar\sX\bar\sY} = D_{\bar\sX\bar\sY} - D_{\bar\sX}D_{\bar\sY} \),
  \( D_{\bar\sX\bar\sY} = D_{\sX\sY} - V^\sigma_{\sX\sY} \), and the requirement of minimal superselection of
  Eq.~\eqref{eq:minimal-superselection}. Finally, the local tomography condition $D_{\sX\sY}=D_{\sX}D_{\sY}$
  for the \Opt \( \Theta \) gives the same expression of the \Lhs.

  Let us now prove the second implication, namely that equality in
  Eq.~\eqref{eq:constraints-on-the-constraints-b} implies \( D_{\bar\sA\bar\sB} = D_{\bar{\sA}} D_{\bar{\sB}}
  \). Expanding \( D_{\bar\sA\bar\sB} \) as \( D_{\sA\sB} - V^\sigma_{\sA\sB} \), using
  Eq.~\eqref{eq:constraints-on-the-constraints-b}, and remembering that the \Opt $\Theta$ is local tomographic
  we get \( D_{\bar\sA\bar\sB} = D_\sA D_\sB - D_\sA V_\sB - D_\sB V_\sA + V^\sigma_\sA V^\sigma_\sB = ( D_\sA
  - V_\sA) ( D_\sB - V_\sB ) = D_{\bar{\sA}} D_{\bar{\sB}} \).
\end{proof}

In general, in a bilocal-tomographic theory two different states $\rho$ and $\sigma$ of the four-partite
system $\sA\sB\sC\sD$ can be discriminated by the following classes of effects
\begin{equation}
  \begin{aligned}
    \Qcircuit @C=1em @R=.7em @! R { 
      & \ustickcool{\sA} \qw & \multimeasureD{1}{a} \\
      & \ustickcool{\sB} \qw & \pureghost{a} \qw \\
      & \ustickcool{\sC} \qw & \multimeasureD{1}{b} \\
      & \ustickcool{\sD} \qw & \pureghost{b} \qw
    }
  \end{aligned}\ ,\quad
  \begin{aligned}
    \Qcircuit @C=1em @R=.7em @! R { 
      & \ustickcool{\sA} \qw & \multimeasureD{1}{c} \\
      & \ustickcool{\sC} \qw & \pureghost{c} \qw \\
      & \ustickcool{\sB} \qw & \multimeasureD{1}{d} \\
      & \ustickcool{\sD} \qw & \pureghost{d} \qw
    }
  \end{aligned}\ ,\quad
  \begin{aligned}
    \Qcircuit @C=1em @R=.7em @! R { 
      & \ustickcool{\sA} \qw & \qw & \multimeasureD{3}{f} \\
      & \ustickcool{\sB} \qw & \multimeasureD{1}{g} \\
      & \ustickcool{\sC} \qw & \pureghost{g} \qw\\
      & \ustickcool{\sD} \qw & \qw & \pureghost{f} \qw
    }
  \end{aligned}\ ,
\end{equation}
where $\Effect{a}$, $\Effect{b}$, $\Effect{c}$, $\Effect{d}$, or $\Effect{f}$ can also be local, \eg
\begin{equation}
  \begin{aligned}
    \Qcircuit @C=1em @R=.7em @! R {
      & \ustickcool{\sA} \qw & \multimeasureD{1}{a} \\
      & \ustickcool{\sB} \qw & \pureghost{a} \qw
    }
  \end{aligned}\ = \ 
  \begin{aligned}
    \Qcircuit @C=1em @R=.7em @! R {
      & \ustickcool{\sA} \qw & \measureD{a_1} \\
      & \ustickcool{\sB} \qw & \measureD{a_2}
    }
  \end{aligned}\ .
\end{equation}

A remarkable feature of maximally bilocal-tomographic theories is given by the following theorem, which
reduces the number of the above classes.

\begin{theorem}\label{th:jellyfish}
  Let $\Theta$ be a maximally bilocal-tomographic theory. Then, for any four-partite system $\sA\sB\sC\sD$ the
  following classes of effects is sufficient in order to discriminate two different states $\rho$ and $\sigma$
  \begin{equation*}
    \begin{matrix}
      \begin{aligned}
        \Qcircuit @C=1em @R=.7em @! R {
          & \ustickcool{\sA} \qw & \qw & \multimeasureD{3}{f} \\
          & \ustickcool{\sB} \qw & \measureD{g} \\
          & \ustickcool{\sC} \qw & \measureD{h} \qw\\
          & \ustickcool{\sD} \qw & \qw & \pureghost{f} \qw }
      \end{aligned}\ \quad &
      \begin{aligned}
        \Qcircuit @C=1em @R=.7em @! R {
          & \ustickcool{\sA} \qw & \measureD{i} \\
          & \ustickcool{\sB} \qw & \multimeasureD{1}{j} \\
          & \ustickcool{\sC} \qw & \pureghost{j} \qw\\
          & \ustickcool{\sD} \qw & \measureD{k} \qw }
      \end{aligned}\ \quad &
      \begin{aligned}
        \Qcircuit @C=1em @R=.7em @! R {
          & \ustickcool{\sA} \qw & \qw & \multimeasureD{2}{l} \\
          & \ustickcool{\sB} \qw & \measureD{m} \\
          & \ustickcool{\sC} \qw & \qw & \pureghost{l} \qw \\
          & \ustickcool{\sD} \qw & \measureD{n} \qw}
      \end{aligned}\ \quad &
      \begin{aligned}
        \Qcircuit @C=1em @R=.7em @! R {
          & \ustickcool{\sA} \qw & \measureD{o} \\
          & \ustickcool{\sB} \qw & \qw & \multimeasureD{2}{p} \\
          & \ustickcool{\sC} \qw & \measureD{q} \qw\\
          & \ustickcool{\sD} \qw & \qw & \pureghost{p} \qw }
      \end{aligned}\ \quad & 
      \begin{aligned}
        \Qcircuit @C=1em @R=.7em @! R {
          & \ustickcool{\sA} \qw & \multimeasureD{1}{r} \\
          & \ustickcool{\sB} \qw & \pureghost{r} \qw \\
          & \ustickcool{\sC} \qw & \multimeasureD{1}{s} \\
          & \ustickcool{\sD} \qw & \pureghost{s} \qw }
      \end{aligned} \\\\
      (i) \quad & (ii) \quad & (iii) \quad & (iv) \quad & (v) \quad 
    \end{matrix}.
  \end{equation*}
\end{theorem}
\begin{proof}
  First, notice that every class (i), (ii), (iii), (iv), (v) spans a linear space of effects of dimension 
  \begin{equation*}
    \begin{matrix}
      D_{\sA\sD} D_\sB D_\sC    \qquad    & 
      D_{\sA} D_{\sB\sC} D_\sD  \qquad    &
      D_{\sA\sC} D_\sB D_\sD    \qquad    &
      D_{\sA} D_{\sB\sD} D_\sC  \qquad    & 
      D_{\sA\sB} D_{\sC\sD}               \\
      (i)   \qquad  &  
      (ii)  \qquad  & 
      (iii) \qquad  & 
      (iv)  \qquad  & 
      (v) 
  \end{matrix}.
  \end{equation*}
  All such linear spaces have a common linear subspace identified by the local effects belonging to the class
  \( \Effect{a}_\sA \otimes \Effect{b}_\sB \otimes \Effect{c}_\sC \otimes \Effect{d}_\sD \). Having this
  subspace dimension \( D_\sA D_\sB D_\sC D_\sD \), we have that the span of all the classes (i)--(v) taken
  together is 
  \begin{equation}\label{eq:classes-dimensions}
      D_{\sA\sD} D_\sB D_\sC    + D_{\sA} D_{\sB\sC} D_\sD   + D_{\sA\sC} D_\sB D_\sD     + D_{\sA} D_{\sB\sD}
        D_\sC   + D_{\sA\sB} D_{\sC\sD} - 4 D_\sA D_\sB D_\sC D_\sD.
  \end{equation}
  We recall here the definition of maximally bilocal-tomographic theory, given in terms of the following
  dimensional relation
  \begin{equation}\label{eq:maxbil}
    D_{\sA\sB\sC} = D_\sA D_\sB D_\sC + D_\sA \tilde D_{\sB\sC} + D_\sB \tilde D_{\sA\sC} +D_\sC \tilde
      D_{\sA\sB}.
  \end{equation}
  Let consider the tripartition $(\sA\sB)\sC\sD$.  Applying the property of maximal bilocal-tomography
  of Eq.~\eqref{eq:maxbil}, we have
  \begin{equation*}
    D_{\sA\sB\sC\sD} = D_{(\sA\sB)\sC\sD} = D_{\sA\sB} D_\sC D_\sD + D_{\sA\sB} \tilde D_{\sC\sD} + D_\sC
      \tilde D_{(\sA\sB)\sD} + D_\sD \tilde D_{(\sA\sB)\sC}.
  \end{equation*}
  By definition, we have $D_{\sA\sB}=D_\sA D_\sB+\tilde D_{\sA\sB}$, and by Eq.~\eqref{eq:maxbil}
  \begin{equation*}
    \tilde D_{(\sA\sB)\sC}=D_{\sA\sB\sC}-D_{\sA\sB}D_\sC=D_\sA \tilde D_{\sB\sC} + D_\sB \tilde D_{\sA\sC}.
  \end{equation*}
  Thus, we conclude
  \begin{multline*}
    D_{\sA\sB\sC\sD} = D_{\sA\sB}(D_\sC D_\sD + \tilde D_{\sC\sD}) + D_\sC \tilde D_{(\sA\sB)\sD} + D_\sD
      \tilde D_{(\sA\sB)\sC} = D_{\sA\sB} D_{\sC\sD} + \\ 
    + D_\sA D_{\sB\sD} D_\sC + D_{\sA\sD}D_\sB D_\sC +D_\sA D_{\sB\sC}D_\sD + D_{\sA\sC}D_\sB D_\sD-4D_\sA
      D_\sB D_\sC D_\sD,
  \end{multline*}
  namely the dimension given by Eq.~\eqref{eq:classes-dimensions}.
\end{proof}
As we will discuss later, the last theorem has important consequences on the notion of entanglement for
Fermionic computation.

\subsection{Quantum theory as an operational probabilistic theory}

It has been shown recently in Ref.~\citen{quit-derivation} that \Qt (in finite dimension) can be regarded as
an \Opt satisfying six properties: the already mentioned causality and local tomography, along with
\emph{perfect distinguishability}, \emph{pure conditioning}, \emph{ideal compression}, and
\emph{purification}.  Thus, all the operational notions introduced in \S\ref{sec:opt}, can be specified in the
case of \Qt. In details, a quantum system $\sA$ is specified by a Hilbert space \( \Hilb{H}_\sA \) with \(
\dim\Hilb{H}_\sA = d_\sA < +\infty \); so \( \Hilb{H}_\sA = \Complexes^{d_\sA} \). The deterministic states
(usually called normalized states) of the system \( \sA \) are the positive semidefinite operators over \(
\Hilb{H}_\sA \) with trace $1$. On the other hand, the linear set of states \( \SetStates_\Reals(\sA) \) is
the whole space \( \Herm(\Hilb{H}_\sA) \) of Hermitian operators over \( \Hilb{H}_\sA \) with dimension \(
D_\sA = d_\sA^2 \).  A non-deterministic preparation test \( \PreparationTest{\rho} = \{ \rho_i \}_{i\in\eta}
\) is a collection of deterministic states \( \{ \tilde\rho_i \} \) along with a collection of probabilities
\( \{ p_i \}_{i\in\eta} \) such that \( \rho_i = p_i \tilde\rho_i \) and \( \sum_{i\in\eta} \Tr[\rho_i] = 1
\).  A deterministic state of \( \sA \) is a rank one projector \( \KetBra{\varphi} \) if it is pure, while it
is a full rank density matrix when it is completely mixed (\eg $I_{\sA}/d_\sA$ with $I_\sA$ the identity
operator on \( \Hilb{H}_\sA \)).  Accordingly the whole set of states \( \SetStates(\sA) \) of system $\sA$ is
the set of all unnormalized density matrices $\rho$, namely \( \rho \geq 0 \), \( \Tr[\rho] \leq 1 \).

Since the effects on \( \sA \) are linear functionals over the set of states we have that the linear space of
effects \( \SetEffects_\Reals(\sA) \) is the space \( \Herm(\Hilb{H}_\sA) \) of Hermitian operators over \(
\Hilb{H}_\sA \). The actual set of effects \( \SetEffects(\sA) \) is made of the positive semidefinite
operators bounded from above by the identity, namely \( \SetEffects(\sA) = \{ P\in\Herm(\Hilb{H}_\sA) \mid P
\geq 0, P \leq I_\sA \} \). An observation test \( \ObservationTest{P} \) is given by a Positive Operator
Valued Measure (\Povm), namely a collection of effects \( \{ P_i \}_{i\in\eta} \) such that \( \sum_{i\in\eta}
P_i = I_{\sA} \).  Again, an atomic effect is simply a rank-one projector.

The probability resulting from the pairing between a state \( \State{\rho} \) and an effect \( \Effect{P} \)
of the system \( \sA \) is given in \Qt\ by the \emph{Born rule}, \ie \( \RBraKet{P}{\rho} \equiv \Tr[\rho
P^\dagger] \).

A transformation \( \Transformation{C} \) between the systems \( \sA \) and \( \sB \) is given by a
\emph{quantum operation}, namely a completely positive trace non-increasing linear map from  \(
\Herm(\Hilb{H}_\sA) \) to \( \Herm(\Hilb{H}_\sB) \). Notice that a quantum operation \( \Transformation{C} \)
always admit the \emph{Kraus decomposition} \( \Transformation{C}(\cdot) = \sum_{i\chi} C_i \cdot C_i^\dagger
\) for suitable bounded operators \( C_i \).  A transformation test \( \Test{C} \subseteq \SetTransf(\sA,\sB)
\) is a collection of quantum operations \( \{ \Transformation{C}_i \}_{i\in\eta} \) such that \(
\sum_{i\in\eta} \Transformation{C}_i \) is a deterministic transformation, namely a \emph{quantum channel},
\ie a trace preserving completely positive map. A unitary transformation---\eg the Schr\"odinger
evolution---is a deterministic test made of a single quantum operation with a Kraus decomposition made of a
single Kraus operator.

\section{The Fermionic Quantum Theory}\label{s:fqt}

In this section we construct an \Opt whose systems are the composition of the so-called \emph{local Fermionic
modes}.  There is not a unique way for realizing a Fermionic \Opt. The one presented here, denoted
\emph{Fermionic Quantum Theory} (\Fqt), stems from simple assumptions on the states/effects of the Fermionic
systems and on the local nature of the Fermionic operations, and is the least constrained theory satisfying
these assumptions. The resulting \Fqt corresponds to a superselected version of the \Qt of qubits with the
superselection rule derived from the consistency of local Fermionic operations in an operational framework.  A
crucial assumption will be that of locality for the Fermionic theory, and it is related to considering the
operator $\F[i]$ as the Kraus operator of an atomic local transformation.

In order to proceed with the construction, first we have to introduce the concept of \emph{Fermionic algebra}.

\subsection{The Fermionic algebra}\label{s:Fermionic-algebra}

The algebra $\mathcal{F}(N)$ of an arbitrary number $N<\infty$ of \emph{local Fermionic modes} ({\Lfm}s) is
generated by Fermionic operators $\{\F[i],\F[i]^\dagger: i \in J_N \}$ with $J_N \coloneqq \{ 1,\ldots,N \}$,
satisfying the \emph{canonical anti-commutation relation} (\Car)
\begin{equation}\label{eq:ccr}
  \{\F[i],\F[i]^\dagger\}=\delta_{ij}I, \qquad \{\F[i],\F[j]\}=0, \qquad\qquad 1\leq i,j \leq N
\end{equation}
Due to the \Car, the positive operators $\F[i]^\dagger\F[i]$ have spectrum $\set S=\{0,1\}$. The operators
$\F[i]$ and $\F[i]^\dagger$ act respectively as \emph{lowering} and \emph{raising} operators for
$\F[i]^\dagger\F[i]$, namely if $\Ket{\Phi}$ is an eigenvector of $\F[i]^\dagger\F[i]$ with eigenvalue $1$
then $\F[i]\Ket{\Phi}$ is an eigenvector with eigenvalue $0$, and $\F[i]^\dag\Ket{\Phi}=0$, while if
$\Ket{\Phi}$ is an eigenvector with eigenvalue $0$ then $\F[i]\Ket{\Phi}=0$ and $\F[i]^\dagger\Ket{\Phi}$ is
an eigenvector with eigenvalue $1$.

The operators $\F[i]^\dagger\F[i]$ form a set of mutually commuting positive operators and we call
\emph{vacuum eigenvector}, denoted $\Ket{\Omega}$, a simultaneous eigenvector with eigenvalue $0$ for all $i$.
A vacuum eigenvector of the Fermionic algebra corresponds to all the {\Lfm}s unoccupied and it is annihilated
by the lowering operators:
\begin{equation*}
  \F[i] \Ket{\Omega} = 0 \qquad \forall i.
\end{equation*}
In general the vacuum $\Ket{\Omega}$ is not unique. However, we can always restrict to the unique case,
corresponding to having a trivial multiplicity, with a vacuum vector space where the field operators act
identically. From now on we will consider the vacuum as unique.

By raising $\Ket{\Omega}$ in all possible ways we get the $2^N$ orthonormal vectors forming the \emph{Fock
basis} in the occupation number representation
\begin{equation}\label{eq:fock}
  \Ket{s_1,\ldots, s_{N}}_\sF \coloneqq ( \F[1]^\dagger)^{s_1}\cdots (\F[N]^\dagger)^{s_N}\Ket{\Omega},
    \qquad s_i\in\{0,1\},
\end{equation}
with $s_i$ corresponding to the occupation number at the $i$th site, \ie the expectation value of the operator
$\F[i]^\dagger\F[i]$. We call \emph{total occupation number} of the vector \( \Ket{s_1,\ldots, s_{N}}_\sF \)
the sum \( \sum_i s_i \).  The linear span of these vectors corresponds to the \emph{anti-symmetric Fock
space} \( \Hilb{F}_N \) of dimension $2^N$.

\subsection{Assumptions}\label{sub:ass}

The assumptions are the following:
\begin{enumerate}[(i)]
  \setlength{\itemsep}{1pt} \setlength{\parskip}{2pt}
  \setlength{\parsep}{1pt} 

\item\label{a:causality} the \Fqt is causal;

\item\label{a:states} the states of $N$ {\Lfm}s are represented by density matrices on on the antisymmetric
  Fock space $\Hilb{F}_N$.

\item\label{a:maps} the transformations on $N$ {\Lfm}s are represented by linear Hermitian preserving maps;

\item\label{a:locphi} the map $\Transformation{X}_i$ with Kraus operators $X_i\coloneqq\F[i] + \F[i]^\dagger $
  is physical;

\item\label{a:locality} for a composite systems $\sA$ made of $N$ {\Lfm}s, transformations with Kraus
  operators in the algebra of field operators $\F[i]$, $\F[i]^\dag$ with $i\in\chi\subset J_N$ are local on
  the subsystem $\sB$ of the {\Lfm}s associated to $\chi$;

\item\label{a:extension} local transformations on a system retain the same Kraus representation when other
  systems are added or discarded;

\item\label{a:pairing} the pairing between states and effects is given by the Born rule \( \RBraKet{a}{\rho}
  \coloneqq \Tr[\rho a] \);

\item\label{a:deteff} on a single \Lfm the pairing with the deterministic effect is represented by
  $\RBraKet{e}{\rho} \coloneqq \Tr[\rho]$.

\end{enumerate}

Notice that since the projection on the vacuum eigenvector has field representation $\KetBra{\Omega} =
\prod_{i=1}^N \F[i] \F[i]^\dagger$, then any state can be written as
\begin{equation*}
  \rho\coloneqq\sum_j K_j \KetBra{\Omega} K_j^\dag= \sum_j K_j\left(\prod_{i=1}^N \F[i] \F[i]^\dagger\right)
    K_j^\dag
\end{equation*}
for some collection of operators
\begin{equation*}
  K_j\coloneqq\sum_{s^{(j)}}\alpha_{s^{(j)}}{\F[1]^\dagger}^{s^{(j)}_1}\cdots{\F[N]^\dagger}^{s^{(j)}_N},\qquad
    \alpha_{s^{(j)}}\in\Complexes.
\end{equation*}
Hence, a state $\rho$ of $N$ {\Lfm}s can be written equivalently as
\begin{equation}\label{eq:state-fields}
  \rho = \sum_{st} \rho_{st} \prod_{i=1}^N {\F[i]^\dagger}^{s_i} \F[i]\F[i]^\dagger \F[i]^{t_i}\qquad
    \rho_{st} \in \Complexes,
\end{equation}
where $s$ is a binary string $s_1 \ldots s_N$ (and similarly for $t$).

Moreover, from assumptions (\ref{a:locphi}), (\ref{a:pairing}), and (\ref{a:deteff}) the following proposition
holds.
\begin{proposition}\label{lem:detx}
  In a system $\sA$ made of $N$ {\Lfm}s for every $i$ the map $\Transformation{X}_i$ is deterministic and
  $\Transformation{X}_i^2=\Transformation{I}$.
\end{proposition}
\begin{proof}
  First notice that from the \Car relations we have $X_i^\dag X_i=X_i^2=(\F[i]^\dag\F[i]+\F[i]\F[i]^\dag)=I$,
  and thus $\Transformation{X}_i^2=\Transformation{I}$.  Moreover, we have
  $\Tr[\Transformation{X}_i(\rho)]=\Tr[X_i\rho X_i^\dag]=\Tr[\rho X_i^\dag X_i]=\Tr[\rho]=1$.
\end{proof}

\subsection{Discarding of a subsystem}

We derive now the simple rule for discarding a subsystem in the \Fqt.  First we need two lemmas that can be
derived by the assumptions.
\begin{lemma}
  Consider a system $\sA=\sB_1\sB_2$ made of $N=N_1+N_2$ {\Lfm}s, and let $\sB_1$ be made of $N_1$ {\Lfm}s
  corresponding to $\chi_1\subset J_N$. Then a transformation $\Transformation{T}\in\SetTransf(\sA)$ is local
  on $\sB_1$ if and only if it can be expressed in terms of Kraus operators belonging to the algebra generated
  by field operators $\F[i]$ and $\F[i]^\dag$ for $i\in\chi_1$.
\end{lemma}
\begin{proof}
  By assumption \eqref{a:maps} a transformation $\Transformation T$ on $\sB_1$ has Kraus operators in the
  algebra of fields $\F[i]$, $\F[i]^\dag$ with $i\in J_{N_1}$.  By assumption \eqref{a:extension}, if we now
  consider the composite system $\sA=\sB_1\sB_2$ the local transformations on $\sB_1$ have Kraus operators in
  the algebra generated by the field operators $\F[i]$, $\F[i]^\dag$ with $i\in J_{N_1}$. On the other hand,
  by assumption \eqref{a:locality}, also the converse is true.
\end{proof}

\begin{lemma}
  The parallel composition of the effect $\Effect{a} \in \SetEffects(\sB_1)$ and the deterministic effect
  $\Effect{e} \in \SetEffects(\sB_2)$ is represented by 
  \begin{equation*}
    \begin{aligned}
      \Qcircuit @C=1em @R=.7em @! R { 
        \multiprepareC{1}{\rho} & \ustickcool{\sB_1} \qw & \measureD{a} \\
        \pureghost{\rho} & \ustickcool{\sB_2} \qw & \measureD{e}
      }
    \end{aligned}\ 
    = \Tr[\rho a], \qquad \forall \State{\rho}\in\SetStates(\sB_2\sB_2).
  \end{equation*}
\end{lemma}
\begin{proof}
  Since in a causal theory \( \Qcircuit @C=1em @R=.7em @!  R { & \measureD{a} } = \Qcircuit @C=1em @R=.7em @!
  R { & \gate{\Transformation{T}} & \measureD{e} } \) for some transformation \( \Transformation{T} \), we
  have
  \begin{equation}
    \begin{aligned}
      \Qcircuit @C=1em @R=.7em @! R { 
        \multiprepareC{1}{\rho} & \ustickcool{\sB_1} \qw & \measureD{a} \\
        \pureghost{\rho} & \ustickcool{\sB_2} \qw & \measureD{e}
      }
    \end{aligned}\ 
    =
    \begin{aligned}
      \Qcircuit @C=1em @R=.7em @! R { 
        \multiprepareC{1}{\rho} & \ustickcool{\sB_1} \qw & \gate{\Transformation{T}} & \ustickcool{\sB_1} \qw
          & \measureD{e} \\
        \pureghost{\rho} & \qw &\ustickcool{\sB_2} \qw & \qw  &\measureD{e}
      }
    \end{aligned}\ .
  \end{equation}
  $\Transformation{T}$, being local on the subsystem $\sB_1$, has Kraus form \( \Transformation{T}(\sigma)
  \coloneqq \sum_i s_i K_i \sigma K_i^\dag \) where $K_i$ is in the algebra of the field operators acting on
  $\sB_1$. By assumption \eqref{a:extension} the map retains the same Kraus expression when extended on a
  system $\sB_1\sB_2$, so
  \begin{equation}
    \begin{aligned}
      \Qcircuit @C=1em @R=.7em @! R { 
        \multiprepareC{1}{\rho} & \ustickcool{\sB_1} \qw & \measureD{a} \\
        \pureghost{\rho} & \ustickcool{\sB_2} \qw & \measureD{e}
      }
    \end{aligned}\ 
    = \sum_j s_j \Tr[ \rho K_j^\dag K_j ] = \Tr[\rho a].
  \end{equation}
\end{proof}

Consider now the system $\sA = \sA_1\ldots\sA_N$ made of $N$ {\Lfm}s, and the bipartition $ \sA = \sB_1 \sB_2
$, corresponding to the disjoint partition $\{\chi_1,\chi_2\}$ of $S$. Since by assumption \eqref{a:causality}
the \Fqt is causal, the marginal state $\sigma$ of system $\sB_1$ for an arbitrary state
$\rho\in\SetStates(\sA)$ is defined by the following implicit equation
\begin{equation*}
  \begin{aligned}
    \Qcircuit @C=1em @R=.7em @! R { 
      \multiprepareC{1}{\rho} & \ustickcool{\sB_1} \qw&\qw\\
      \pureghost{\rho} & \ustickcool{\sB_2} \qw & \measureD{e}
    }
  \end{aligned}
  =
  \begin{aligned}
    \Qcircuit @C=1em @R=.7em @! R { 
      \prepareC{\sigma} & \ustickcool{\sB_1} \qw&\qw
    }
  \end{aligned}
  \quad\Leftrightarrow\quad
  \begin{aligned}
    \Qcircuit @C=1em @R=.7em @! R { 
      \multiprepareC{1}{\rho} & \ustickcool{\sB_1} \qw & \measureD{a} \\
      \pureghost{\rho} & \ustickcool{\sB_2} \qw & \measureD{e}
    }
  \end{aligned}
  =
  \begin{aligned}
    \Qcircuit @C=1em @R=.7em @! R { 
      \prepareC{\sigma} & \ustickcool{\sB_1} \qw & \measureD{a}
    },
  \end{aligned}
\end{equation*}
for any effect $a\in\SetEffects(\sB_1)$ on the complementary system of $\sB_2$. In formula we write
\begin{equation}\label{eq:partial-trace}
  \Tr[ \sigma a ]: =\Tr[\rho a],\;\forall a\in\SetEffects(\sB_1).
\end{equation}
Let $ \State{\rho} \in \SetStates(\sB_1\sB_2)$ be as in Eq.~\eqref{eq:state-fields}, we can perform the
following swapping of the field operators 
\begin{equation*}
  \rho = \sum_{st} \rho_{st} \prod_{i=1}^N {\F[i]^\dagger}^{s_i} \F[i]\F[i]^\dagger \F[i]^{t_i} = 
  \sum_{st} (-1)^{f(s,t)} \rho_{st} \prod_{k\in\chi_2} ({ \F[k]^\dagger }^{ s_k } \F[k] \F[k]^\dagger \F[k]^{
    t_{k} } ) \prod_{i\in\chi_1} ( { \F[i]^\dagger }^{ s_i } \F[i] \F[i]^\dagger \F[i]^{t_{i}} ),
\end{equation*}
where $f(s,t)$ is the function evaluating the number of swaps needed to perform the reordering, which is given
by
\begin{equation*}
  f(s,t) \coloneqq \sum_{k\in\chi_2} (s_k \oplus t_k) 
  \ \ \Sigma_{\substack{i\in\chi_1\\i<k}} 
  (s_i
    \oplus t_i) .
\end{equation*}

The \Rhs of Eq.~\eqref{eq:partial-trace} then becomes
\begin{multline*}
  \sum_{st}\rho_{st}\Tr \left[\prod_{i=1}^{N} ( { \F[i]^\dagger }^{ s_i } \F[i] \F[i]^\dagger
     \F[i]^{ t_{i} } )\ a\right] = \\
   \sum_{st} (-1)^{f(s,t)} \rho_{st} \Tr \left[\ \prod_{i\in\chi_1} ( { \F[i]^\dagger }^{ s_i } \F[i]
     \F[i]^\dagger \F[i]^{t_{i}} )\ a\ \prod_{k\in\chi_2} ({ \F[k]^\dagger }^{ s_k } \F[k] \F[k]^\dagger
     \F[k]^{ t_{k} } ) \right] = \\
   \sum_{st} (-1)^{f(s,t)} \rho_{st} \Tr \left[\ \prod_{i\in\chi_1} ( { \F[i]^\dagger }^{ s_i } \F[i]
     \F[i]^\dagger \F[ i]^{ t_{i} })\ a\ \prod_{k\in\chi_2} \delta_{s_k,t_k}\ \right].
\end{multline*}
Since $f(s,t) = 0$ whenever for $k\in\chi_2$ it happens $s_k = t_k$, the previous equation shows that the
marginal state on subsystem $\sB_1$ of a state $\rho\in\SetStates(\sB_1\sB_2)$ is given by
\begin{equation}\label{eq:discarding}
  \sigma \coloneqq \Tr_{ \sB_2 } [ \rho ] = \Tr_{\sB_2} \left[\ \sum_{s,t} \rho_{st} \prod_{i=1}^N
    { \F[i]^\dagger }^{ s_{i} } \F[i] \F[i]^\dagger \F[i]^{ t_{i} }\ \right] = \!\!\!\!\! \sum_{
    \substack{ s,t \textrm{ with} \\ s_k=t_k,\ k\in\chi_2}} \!\!\!\!\! \rho_{st} 
  \prod_{i\in\chi_1} { \F[i]^\dagger }^{ s_{i} } \F[i] \F[i]^\dagger \F[i]^{ t_{i} },
\end{equation}
namely it is obtained by dropping all terms that contain an odd number of field operators in any of the
{\Lfm}s in $\sB_2$, while in the remaining terms one erases the field operators in $\sB_2$.

\subsection{Derivation of the parity superselection rule}\label{s:parity-superselection}

In the following we will show that the Wigner parity superselection rule
\cite{streater1964pct,weinberg1996quantum} can be derived operationally from Postulates \eqref{a:locphi} and
\eqref{a:locality}. 

\begin{theorem}\label{t:fqt-kraus}
  Every transformation between \( N \) {\Lfm}s is operationally equivalent to a map where each Kraus operator
  is a combinations of products of either odd or even numbers of field operators.
\end{theorem}
\begin{proof}
  Let us take an arbitrary transformation \( \Transformation{T} \in \SetTransf_\Reals(\sA, \sB) \) with \( \sA
  \), \( \sB \) $N$-\Lfm systems.  Since by assumption \eqref{a:maps} $\Transformation{T}$ is hermitian
  preserving, it can be written as the difference between two CP maps, hence, for an arbitrary $\rho$,
  \( \Transformation{T}(\rho) = \sum_i s_i K_i \rho {K_i}^\dagger \), where \( K_i \) are Kraus operators, and
  \( s_i = \pm 1 \) for every $i$.  Every $K_i$ can be decomposed as \( K_i = E_i + O_i \) with $E_i, O_i \in
  \mathcal{L}(\Complexes^{2^N})$, and \( E_i \) and $O_i$ being the part of $K_i$ containing only
  superposition of an even and odd number of field operators, respectively.  Thus, we have \(
  \Transformation{T} = \sum_i s_i( E_i \cdot {E_i}^\dagger + O_i \cdot {O_i}^\dagger + E_i \cdot {O_i}^\dagger
  + O_i \cdot {E_i}^\dagger) \).  We want to show that \( \Transformation{T} \) is equivalent to the map \(
  \tilde{\Transformation{T}} \coloneqq \sum_i s_i ({E}_i \cdot {E_i}^\dagger + {O}_i \cdot {O_i}^\dagger) \),
  namely for every ancillary system \( \sC \) made of $M$ {\Lfm}s
  \begin{equation*}
    \begin{aligned}
      \Qcircuit @C=1em @R=.7em @! R {
        \multiprepareC{1}{\rho} & \qw \poloFantasmaCn{\sA} & \gate{\Transformation{T}} & \qw
          \poloFantasmaCn{\sB} 
          & \multimeasureD{1}{a} \\
        \pureghost{\rho} & \qw & \qw \poloFantasmaCn{\sC} & \qw & \pureghost{a} \qw
      }
    \end{aligned}
    =
    \begin{aligned}
      \Qcircuit @C=1em @R=.7em @! R {
        \multiprepareC{1}{\rho} & \qw \poloFantasmaCn{\sA} & \gate{\tilde{\Transformation{T}}} & \qw
          \poloFantasmaCn{\sB} 
          & \multimeasureD{1}{a} \\
        \pureghost{\rho} & \qw & \qw \poloFantasmaCn{\sC} & \qw & \pureghost{a} \qw
      }
    \end{aligned}\,, \quad
      \forall\State{\rho}\in\SetStates(\sA\sC),\Effect{a}\in\SetEffects(\sB\sC).
  \end{equation*}
  Using assumption \eqref{a:pairing} the previous relation is equivalent to 
  \begin{equation}\label{eq:proof}
    \sum_i s_i \Tr[\ (E_i \rho {O_i}^\dagger + O_i \rho {E_i}^\dagger)\ a^\dagger\ ] =0,
  \end{equation}
  for every \( \State{\rho}\in\SetStates(\sA\sC) \), \( \Effect{a}\in\SetEffects(\sB\sC) \) and every
  ancillary system \( \sC \).
  
  In order to prove Eq.~\eqref{eq:proof} we consider the physical map---by assumption \eqref{a:locphi}---\(
  \SetTransf(\sC,\sD) \ni \Transformation{X}_j\), $j$ denoting a Fermionic subsystem belonging to \( \sC \),
  and \( \sD \) a Fermionic system made of $M$ {\Lfm}s too. Being \( \Transformation{X}_j \) and \(
  \Transformation{T} \) two transformations acting on different subsystems, by
  Eq.~\eqref{eq:local-transformations-commute} their sequential composition shall commute; \ie for every
  ancillary system $\sE$, for every state \( \State{\sigma} \in \SetStates(\sA\sC\sE) \), and for every effect
  \( \Effect{b} \in \SetEffects(\sB\sD\sE) \)
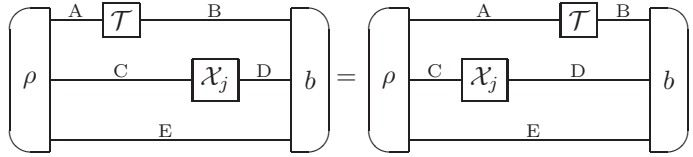
  \begin{equation*}
    \begin{aligned}
      \Qcircuit @C=1em @R=.7em @! R {
        \multiprepareC{2}{\rho} & \qw \poloFantasmaCn{\sA} & \gate{\Transformation{T}} & \qw
          & \qw \poloFantasmaCn{\sB}  & \qw 
          & \multimeasureD{2}{b} \\
        \pureghost{\rho} & \qw & \qw \poloFantasmaCn{\sC} & \qw & \gate{\Transformation{X}_j} & \qw
          \poloFantasmaCn{\sD} & \pureghost{b} \qw \\
        \pureghost{\rho} & \qw & \qw & \qw \poloFantasmaCn{\sE} & \qw & \qw & \pureghost{b} \qw
      }
    \end{aligned}
    =
    \begin{aligned}
      \Qcircuit @C=1em @R=.7em @! R {
        \multiprepareC{2}{\rho} & \qw & \qw \poloFantasmaCn{\sA} & \qw & \gate{\Transformation{T}} & \qw
          \poloFantasmaCn{\sB} & \multimeasureD{2}{b} \\
        \pureghost{\rho} & \qw \poloFantasmaCn{\sC} & \gate{\Transformation{X}_j} & \qw
          & \qw \poloFantasmaCn{\sD}  & \qw & \pureghost{b} \qw \\
        \pureghost{\rho} & \qw & \qw & \qw \poloFantasmaCn{\sE} & \qw & \qw & \pureghost{b} \qw
      }
    \end{aligned}\,.
  \end{equation*}
  Consider the case where $\sD$ is the system made of $0$ {\Lfm}s, \ie the ancillary system is the trivial
  system $\sI$. Then we have by assumption \eqref{a:pairing} that a necessary condition for the commutation of
  the maps \( \Transformation{T} \) and \( \Transformation{X}_j \) is given by
  \begin{multline*}
    \sum_i s_i \Tr[(X_j E_i \rho {E_i}^\dagger X_j^\dagger + X_j O_i \rho {O_i}^\dagger X_j^\dagger + X_j E_i
      \rho {O_i}^\dagger X_j^\dagger + X_j O_i \rho {E_i}^\dagger X_j^\dagger)b ] = \\
    = \sum_i s_i \Tr[\ ( E_i X_j \rho X_j^\dagger {E_i}^\dagger + O_i X_j \rho X_j^\dagger {O_i}^\dagger + E_i
      X_j \rho X_j^\dagger {O_i}^\dagger + O_i X_j \rho X_j^\dagger {E_i}^\dagger )b ]\\
    \forall \State{\rho}\in\SetStates(\sA\sC\sE), \forall\Effect{b}\in\SetEffects(\sB\sD\sE).
  \end{multline*}
  Since the Kraus operators \( E_i \), \( O_i \) contain respectively an even and an odd number of field
  operators, the anti-commutation relations for the fields and the invariance of the trace under cyclic
  permutation give us
  \begin{equation*}
    \sum_i s_i \Tr[\ (E_i \rho {O_i}^\dagger + O_i \rho {E_i}^\dagger)X_j^\dagger bX_j ] = 0
  \end{equation*}
  If we now choose $b=X_j a X_j^\dag $ for an arbitrary $a\in\SetEffects(\sA\sC)$, by proposition
  \ref{lem:detx} we obtain
  \begin{equation}\label{eq:asddsa}
    \sum_i s_i \Tr[(E_i \rho {O_i}^\dagger + O_i \rho {E_i}^\dagger)a] =0 \qquad \forall
      \State{\rho}\in\SetStates(\sA\sC), \Effect{a}\in\SetEffects(\sB\sC),
  \end{equation}
  namely Eq.~\eqref{eq:proof}. We conclude therefore that the compatibility condition of commutation between
  local transformation implies the thesis.
\end{proof}

The previous theorem allows us to consider the transformations with each Kraus operator involving only an even
or an odd number of field operators as the representatives of the equivalence class they belong to. This fact
allows us to prove the following corollary.

\begin{corollary}\label{cor:effects}
  Effects of the \Fqt are positive operators made of products of an even number of field operators.
\end{corollary}
\begin{proof}
  Since in a causal theory we have \( \Qcircuit @C=1em @R=.7em @!  R { & \measureD{a} } = \Qcircuit @C=1em
  @R=.7em @! R { & \gate{\Transformation{T}} & \measureD{e} } \) for some transformation \( \Transformation{T}
  \), every effect can be written as \( a = \sum_i s_i {E_i}^\dagger E_i + s_i {O_i}^\dagger O_i \), namely an
  operator involving only products of even number of field operators.
\end{proof}

\begin{lemma}
  The even part of a state $\rho$ is a density matrix.
\end{lemma}
\begin{proof}
  By assumption \eqref{a:states} a state of $N$ {\Lfm}s is a positive operator $\rho$ on \( \Hilb{F}_N \),
  then it can be expressed as $\rho=X^\dag X$. Writing $X=E+O$, with $E$ combination of even products of field
  operators and $O$ combination of odd products, we have $\rho=E^\dag E+O^\dag O+E^\dag O+O^\dag E$. Finally,
  the even part of $\rho$ is given by $E^\dag E+O^\dag O\geq0$, which is positive.
\end{proof}

\begin{proposition}\label{prop:superselection}
  States of \Fqt satisfy the parity superselection rule.
\end{proposition}
\begin{proof}
  Consider the state \( \State{\rho} = \sum_j E_j + O_j \), and its even part \( \State{\rho_E} \coloneqq
  \sum_j E_j \) with \( E_i \) and \( O_i \) made of linear combinations of an even and an odd number of field
  operators, respectively. Since \( \Tr[\ O_j a \ ] = 0 \), due to $a$ being made of products of an even
  number of field operators (see Corollary~\ref{cor:effects}), we have that \( \State{\rho} \) is
  operationally equivalent to \( \State{\rho_E} \), that is for every effect \( \Effect{a} \), \(
  \RBraKet{a}{\rho} = \RBraKet{a}{\rho_E} \).  Hence it is not restrictive to consider only the states
  represented by density matrices that are linear combinations of products of even number of field operators,
  as representatives of the resulting equivalence classes of states.  One can now decompose the Fock space \(
  \Hilb{F}_N \) in the direct sum
  \begin{equation}\label{eq:hilbert-sum}
    \Hilb{F}_N=\Hilb{F}_N^{0} \oplus \Hilb{F}_N^{1},
  \end{equation}
  where $\Hilb{F}_N^0$ and $\Hilb{F}_N^1$ are the eigenspaces of the parity operator
  \begin{equation}
    P=\tfrac{1}{2}(I+\prod_{i=1}^{n}(\F_i\F^\dag_i-\F^\dag_i\F_i))  
  \end{equation} 
  corresponding to the eigenvalues $p=0,1$---\ie corresponding to an even/odd total occupation number. We
  conclude that every state---being represented by a combination of products of an even number of
  fields---commutes with \( P \), thus it has a well defined parity, \ie states satisfy the parity
  superselection rule.
\end{proof}

\begin{corollary}
  The vacuum state $\KetBra{\Omega}$ is physical.
\end{corollary}
\begin{proof}
  Being \( \KetBra{\Omega} = \prod_i \F[i]\F[i]^{\dagger} \) a state with an even number of field operators,
  it is valid state of the \Fqt.
\end{proof}

Finally, since effects of a system of $N$ {\Lfm}s are linear combinations of even products of field operators,
they commute with the parity operator, too. This allows us to derive the parity superselection rule also for
effects.

\begin{corollary}
  Effects of \Fqt satisfy the parity superselection rule.
\end{corollary}

\subsection{Set of states, effects and transformations}

In the following we will analyze the consequences of the parity prescription on the states, effects, and
transformations of \Fqt.

For this purpose, we remind that if $\sA$ is a $N$ qubits system, then the linear spaces of states and effects
correspond to the set of $2^N\times 2^N$ Hermitian matrices
\begin{equation}\label{eq:qubits-theory-set-of-states}
  \SetStates_{\Reals}(\sA)=\SetEffects_{\Reals}(\sA)= \operatorname{Herm}(
    (\Complexes^2)^{\otimes N} )
\end{equation} 
and the dimension of the set of states and effects is
\begin{equation}\label{eq:quantum-dim2}
  D_{\sA}=d_{\sA}^2=2^{2N},
\end{equation}
with $d_\sA = 2^N$ the Hilbert dimension of $N$ qubits. On the other hand, a system of $N$ {\Lfm}s must obey
the parity superselection rule, which forbids any pure state corresponding to a superposition of vectors
belonging to $\Hilb{F}_N^0$ and $\Hilb{F}_N^1$, \ie pure states are given by projections on superpositions of
Fock vectors with total occupation numbers equal modulo 2.  Hence the elementary system---the one-\Lfm---has
only the pure states $\KetBra{0},\KetBra{1}$, thus corresponding to the classical \emph{bit} (indeed the Fock
vectors $\Ket{0}$ and $\Ket{1}$ belong to \( \Hilb{F}_1^0 \) and \( \Hilb{F}_1^1 \), respectively, and then
one cannot consider their superpositions).  In general, for a system \( \sA \) of $N$ {\Lfm}s we can identify
two disjoint sectors with different parity in the linear sets of states and effects: $\SetStates_{\Reals}(\sA)
=\SetEffects_{\Reals}(\sA) = \Herm(\Hilb{F}_N^0) \oplus \Herm(\Hilb{F}_N^1)$. Since \( \dim \Hilb{F}_N^0 =
\dim \Hilb{F}_N^1 = 2^{N-1} \) we have \( \dim \Herm(\Hilb{F}_N^i) = 2^{2(N-1)} \). Being the dimension of the
linear space of states of \( N-1 \) qubits exactly \( 2^{2(N-1)} \), we have that each parity sector of the
linear set of states of \( N \) {\Lfm}s is isomorphic to that of \( N-1 \) qubits, making \(
\SetStates_\Reals(\sA)=\SetEffects_\Reals(\sA) \) equivalent to the direct sum of two $N-1$ qubit state
spaces, with
\begin{equation}\label{eq:Fermionic-dimension}
  D_{N\,\Lfm\textrm{s}} = 2 D_{N-1\,\textsc{qubit}\textrm{s}} = \frac{1}{2} D_{N\,\textsc{qubit}\textrm{s}} =
    2^{2N-1}. 
\end{equation}
A general element of $\SetStates_{\Reals}(\sA) = \Herm(\Hilb{F}_N^0) \oplus \Herm(\Hilb{F}_N^1)$ has a block
diagonal form, that characterizes also the actual sets of states and effects: reordering the basis of the Fock
space $\Hilb{F}_N$ in such a way that all the even vectors precede all the odd ones, one has that for every
state \( \rho \in \SetStates(\sA)\) and every effect $a\in\SetEffects(\sA)$
\begin{align*}
  & \rho=\left(
    \begin{array}{l|l}
      \rho_0&0\\
      \hline 0&\rho_1
    \end{array}
  \right),&& \rho_i\geq0,\quad\Tr[\rho_0]+\Tr[\rho_0]\leq 1,\\
  & a=\left(\begin{array}{l|l}
      a_0&0\\
      \hline 0&a_1
    \end{array}
  \right),&& 0\leq a_i\leq I_{2^{N-1}},
\end{align*}
corresponding to \(\SetStates(\sA)= \operatorname{Conv}[\ (\SetStates(\sA) \cap \Herm(\Hilb{F}_N^0)) \cup
(\SetStates(\sA) \cap \Herm(\Hilb{F}_N^1))\ ] \) and \(\SetStates(\sA)= \operatorname{Conv}[\
(\SetEffects(\sA) \cap \Herm(\Hilb{F}_N^0)) \cup (\SetEffects(\sA) \cap \Herm(\Hilb{F}_N^1))\ ] \), with \(
\operatorname{Conv}(X) \) representing the convex hull of the set $X$.

Notice that, thanks to the definition of Eq.~\eqref{eq:fock} the Fock space $\Hilb{F}(N)$ is isomorphic to a
$N$-qubit Hilbert space, by the trivial identification of the occupation number basis $\Ket{s_1,\ldots,s_N}_F$
with the qubit computational basis $\Ket{s_1,\ldots,s_N}_Q$ of eigenvectors of the Pauli matrices $\sigma^z_i$
with $1\leq i\leq N$. Hence the two parity sectors \( \SetStates(\sA) \cap \Herm(\Hilb{F}_N^i) \) are actually
isomorphic to the \( (N-1) \)-qubit states set, with pure states given by the rank one projectors \(
\KetBra{\Psi} \) with \( \Ket{\Psi} \) normalized superposition of Fock vectors belonging to \( \Hilb{F}_N^i
\), while the two sectors \( \SetEffects(\sA) \cap \Herm(\Hilb{F}_N^i) \) are isomorphic to the \( (N-1)
\)-qubit effects set, whose atomic elements coincide with pure states.

\begin{proposition}\label{prop:CP-maps}
  Let \( \sA \), \( \sB \) be two $N$ {\Lfm}s systems. Then, transformations from \( \sA \) to \( \sB \) are
  CP maps from \( \SetStates(\sA) \) to \( \SetStates(\sB) \). 
\end{proposition}
\begin{proof}
  Since the parity superselection implies the presence of the two parity sectors \( \SetStates(\sA) \cap
  \Herm(\Hilb{F}_N^i) \) with $i=0,1$, we have that an arbitrary transformation \( \Transformation{T} \in
  \SetTransf(\sA, \sB) \) can be written as \( \Transformation{T} = \Transformation{T}_{00} +
  \Transformation{T}_{01} + \Transformation{T}_{10} + \Transformation{T}_{11} \), with 
  \begin{align*}
    \Transformation{T}_{xy}:\quad&\SetStates(\sA) \rightarrow \SetStates(\sB)\\
    & \State{\rho} \mapsto 
      \begin{cases}
        \Transformation{T}_{xy}\State{\rho} \in \SetStates(\sA) \cap \Herm(\Hilb{F}_N^y)&\text{if }
          \State{\rho}\in\SetStates(\sA) \cap \Herm(\Hilb{F}_N^x)\\
        0 & \text{if }\State{\rho}\in\SetStates(\sA) \cap \Herm(\Hilb{F}_N^{\bar x})
      \end{cases}.
  \end{align*}
  Since \( \SetStates(\sA) \cap \Herm(\Hilb{F}_N^0) \sim \SetStates(\sA) \cap \Herm(\Hilb{F}_N^1) \sim
  \SetStates(\sB ) \) with $\sB$ a $N-1$ qubits system, we have that all the \( \Transformation{T}_{xy} \)
  with \( x,y=0,1 \) are actually quantum maps from \( N-1 \) to \( N-1 \) qubits, \ie CP maps.
\end{proof}

Since we are admitting the no-restriction hypothesis, all the transformations with Kraus operators being
superpositions of products of either an even number or an odd number of field operators belong to the theory
(since they are admissible).  Finally, every admissible transformation can be dilated to a single-Kraus one
thanks to the following proposition.
\begin{proposition}
  Every multi-Kraus transformation can be dilated to a single Kraus one.
\end{proposition}
\begin{proof}
  Let \( \Transformation{T} = \sum_{i\in\eta_\mathrm{E}} E_i \cdot E_i^\dagger + \sum_{i\in\eta_\mathrm{O}}
  O_i \cdot O_i^\dagger  \) be a transformation of the $N$-\Lfm system $\sA$, with $E_i$, $O_i$ even and odd
  Kraus operators respectively. We want to show that we can find \( \tilde{\Transformation{T}} = \tilde{T}
  \cdot \tilde{T}^\dagger \) acting on $\sA\sB$ with $\sB$ a $M$-\Lfm system and a state $\State{\sigma} \in
  \SetStates(\sB)$ such that 
  \begin{equation}\label{eq:dilation}
    \begin{aligned}
      \Qcircuit @C=1em @R=.7em @! R {
        & \ustickcool{\sA}\qw & \gate{\Transformation{T}} & \ustickcool{\sA}\qw &\qw 
      } 
    \end{aligned}
    =
    \begin{aligned}
      \Qcircuit @C=1em @R=.7em @! R {
        & \ustickcool{\sA}\qw & \multigate{1}{\tilde{\Transformation{T}}} & \ustickcool{\sA}\qw &  \qw   \\
        \prepareC{\sigma} & \ustickcool{\sB}\qw  & \ghost{\Transformation{T}} & \ustickcool{\sB}\qw & \measureD{e}
      } 
    \end{aligned}.
  \end{equation}
  In a $K$-\Lfm system there are $2^{K-1}$ Fock vectors of Eq.~\eqref{eq:fock} involving an even number of
  fields (as well as $2^{K-1}$ involving an odd number of them). An even Fock vector \( \Ket{e_i} \) can be
  written as \( \tilde{E}_i \Ket{\Omega} \), with $E_i$ an operator involving an even number of fields
  (similarly we have \( \Ket{o_i} = \tilde{O}_i \Ket{\Omega} \) for the odd ones).  We set $\sB$ to be a
  $M$-\Lfm system with $M \coloneqq \max [ \operatorname{ceiling} (\log_2 |\eta_\mathrm{E}|) ,
  \operatorname{ceiling}( \log_2 |\eta_\mathrm{O}|) ] + 1$: in this way $\sB$ is just big enough to allocate a
  number of even and odd Fock vectors equals respectively to the number of even and odd Kraus operators
  appearing in $\Transformation{T}$.  Moreover, let $\State{\sigma}$ be the vacuum state of $M$ {\Lfm}s, then
  a dilation of \( \Transformation{T} \) is given by the transformation \( \tilde{\Transformation{T}} \) with
  even single-Kraus
  \begin{equation}
    \tilde{T} = \sum_{i\in\eta_{\mathrm{E}}} E_i\tilde{E}_i + \sum_{i\in\eta_{\mathrm{O}}} O_i\tilde{O}_i,
  \end{equation}
  where \( \tilde{E}_i \) and \( \tilde{O}_i \) are the even and odd field operators defining the even and odd
  orthonormal Fock vectors for the system $\sB$. Let us show the equality of Eq.~\eqref{eq:dilation}, namely
  \begin{equation}\label{eq:incubo}
    \begin{aligned} 
      \Qcircuit @C=1em @R=.7em @! R { 
        \multiprepareC{1}{\rho} & \qw & \ustickcool{\sC}\qw & \qw &\qw \\
        \pureghost{\rho} & \ustickcool{\sA}\qw & \gate{\Transformation{T}} & \ustickcool{\sA}\qw &\qw 
      } 
    \end{aligned} 
    = 
    \begin{aligned} 
      \Qcircuit @C=1em  @R=.7em @! R { 
        \multiprepareC{1}{\rho} & \qw & \ustickcool{\sC}\qw & \qw &\qw \\
        \pureghost{\rho} & \ustickcool{\sA}\qw & \multigate{1}{\tilde{\Transformation{T}}} &
          \ustickcool{\sA}\qw & \qw \\ 
        \prepareC{\sigma} & \ustickcool{\sB}\qw  & \ghost{\Transformation{T}} & \ustickcool{\sB}\qw &
          \measureD{e} 
      } 
    \end{aligned},
  \end{equation}
  for an arbitrary system $\sC$ of $P$ {\Lfm}s and an arbitrary state \( \State{\rho} \in \SetStates(\sC\sA)
  \).  The \Lhs of Eq.~\eqref{eq:incubo} is given by
  \[
   \sum_{i\in\eta_\mathrm{E}} E_i \rho E_i^\dagger + \sum_{i\in\eta_\mathrm{O}} O_i \rho O_i^\dagger
  \]
  On the other hand, being \( \tau \coloneqq \State{\rho} \otimes \State{\sigma} = \rho \prod_{i=1}^M
  \tilde{\F[i]}\tilde{\F[i]}^\dagger \), with $\tilde{\F[i]}, \tilde{\F[i]}^\dagger$ the field operators on
  the subsystem $\sB$, the \Rhs of Eq.~\eqref{eq:incubo} is 
  \begin{multline*}
    \Tr_{\sB} \left[\ \tilde{T}\ \tau\ \tilde{T}^\dagger\ \right] = \sum_{j,k\in\eta_\mathrm{E}} \Tr_{\sB}
      \left[\ E_j \tilde{E}_j\ \tau\ \tilde{E}_k^\dagger E_k^\dagger\ \right] + \sum_{j,k\in\eta_\mathrm{O}}
      \Tr_{\sB} \left[\ O_j \tilde{O}_j\ \tau\ \tilde{O}_k^\dagger O_k^\dagger\ \right] + \\ 
    + \sum_{j\in\eta_\mathrm{E}} \sum_{k\in\eta_\mathrm{O}} \left\{ \Tr_{\sB} \left[\ E_j \tilde{E}_j\ \tau \
      \tilde{O}_k^\dagger O_k^\dagger\ \right]  + \Tr_{\sB} \left[\ O_k \tilde{O}_k\ \tau \ \tilde{E}_j^\dagger
      E_j^\dagger\ \right] \right\} = \\
    \sum_{j,k\in\eta_\mathrm{E}} \Tr_{\sB} \left[\ E_j  \rho E_k^\dagger \ \ \tilde{E}_k^\dagger
      \prod_{i=1}^M \tilde{\F[i]}\tilde{\F[i]}^\dagger \tilde{E}_j \ \right] + \sum_{j,k\in\eta_\mathrm{O}}
      \Tr_{\sB} \left[\ O_j \rho O_k^\dagger\ \ \tilde{O}_k^\dagger \prod_{i=1}^M
      \tilde{\F[i]}\tilde{\F[i]}^\dagger \tilde{O}_j  \ \right] + \\ 
   - \sum_{j\in\eta_\mathrm{E}} \sum_{k\in\eta_\mathrm{O}} \left\{ \Tr_{\sB} \left[\ E_j \rho O_k^\dagger \ \
     \tilde{O}_k^\dagger \prod_{i=1}^M \tilde{\F[i]}\tilde{\F[i]}^\dagger \tilde{E}_j \ \right]  + \Tr_{\sB}
     \left[\ O_k \rho E_j^\dagger \ \ \tilde{E}_j^\dagger \prod_{i=1}^M \tilde{\F[i]}\tilde{\F[i]}^\dagger
     \tilde{O}_k \ \right] \right\}.
  \end{multline*}
  Due to the orthogonality relation between the Fock vectors, Eq.~\eqref{eq:discarding} shows that the
  previous equation is equal to Eq.~\eqref{eq:incubo}. Notice that with a similar procedure we could have
  dilated $\Transformation{T}$ to an \emph{odd} single-Kraus transformation.
\end{proof}

\section{Informational features}\label{sec:informational-features}

In this section we derive the consequences of the parity superselection on the structure of \Fqt. We will
explore the tomography of Fermionic states (which results to be non-local), the properties of Fermionic
entanglement (which exhibits differences with respect to the quantum case), and some issues regarding the
computation in the \Fqt.

First of all we introduce the Jordan-Wigner isomorphism, which will be useful to compare \Fqt with \Qt and to
address the issue of simulation.

\subsection{The Jordan-Wigner map}\label{sec:JW}

Thanks to Eq.~\eqref{eq:fock}, the Fock space $\Hilb{F}_N$ and Hilbert space of $N$ qubits
${(\Complexes^2)}^{\otimes N}$ are isomorphic. A simple way to map unitarily an orthonormal basis of the
former to an orthonormal basis of the latter is
\[
  \Ket{s_1, \ldots, s_N}_\mathrm{F} \overset{U}{\rightarrow} \Ket{s_1, \ldots, s_N}_\mathrm{Q},
\]
where \( \Ket{s_1, \ldots, s_N}_\mathrm{Q} \) is the joint eigenvector of the qubit operators \( \sigma^z_j \)
with \( j=1,\ldots,N \).  Notice that such an encoding necessarily depends on the chosen ordering for the
{\Lfm}s in Eq.\eqref{eq:fock}. Indeed, had we chosen a different ordering $\pi\in S_N$ in Eq.~\eqref{eq:fock}
we would have got the Fock vectors 
\[
  \Ket{s_1,\ldots,s_{N}}_\sF^\pi \coloneqq (\F[\pi(1)]^\dagger)^{s_{\pi(1)}} \cdots
    (\F[\pi(N)]^\dagger)^{s_{\pi(N)}} \Ket{\Omega} \equiv (-1)^{\operatorname{sign}(\pi)}
    \Ket{s_1,\ldots,s_{N}}_\sF, 
\]
and the new unitary map would have been
\[
  \Ket{s_1, \ldots, s_N}^\pi_\mathrm{F} \overset{U^\pi}{\rightarrow} \Ket{s_1, \ldots, s_N}_\mathrm{Q}.
\]

For a given ordering $\pi$ the map $U^\pi$ induces a *-algebra isomorphism between the \Car algebra of the
fields and the algebra of the Pauli matrices \( \{ \sigma_j^\alpha \} \) known as \emph{Jordan-Wigner
transform} (\Jwt).  For example, for a given ordering $\pi$ the \Jwt gives
\begin{equation*}
  \F[i] \rightarrow J_\pi (\F[i]) = \prod_{k=\pi(1)}^{\pi(i-1)} \sigma_k^z
    \sigma^{-}_{\pi(i)},\qquad\text{with } \sigma^{\pm}_k \coloneqq \frac{\sigma^x_k \pm i\sigma^y_k}{2}. 
\end{equation*}
From the previous equation we notice that under \Jwt a single \Lfm field operator is in general mapped to a
many qubits operator. This is a general property of the \Jwt regardless the number of {\Lfm}s involved. For
instance, the 2-\Lfm field operator $\F[i]^\dagger \F[j]$ is mapped under a \Jwt to
\begin{equation*}
  \F[i]^\dagger \F[j] \rightarrow J_\pi (\F[i]^\dagger \F[j]) = \sigma^+_{\pi(i)}
    \prod_{\pi(i)<k<\pi(j)}\sigma^z_k \sigma^-_{\pi(j)},
\end{equation*}
namely the corresponding qubit operator involves more than two subsystems, the only exception when the chosen
ordering gives \( \pi(j) = \pi(i) + 1 \).

In the following we will denote by $J$ the \Jwt representation corresponding to the trivial ordering
permutation.  Under the trivial ordering the Pauli matrices can be expressed in terms of the Fermionic
operators $\F[1],\ldots,\F[N]$ as follows
\begin{align}
  \label{eq:sigma-x}  & \sigma_i^{x} = \sigma^z_{1} \cdots \sigma^z_{i-1} J( \F[i] + \F[i]^\dagger ), \\ 
  \label{eq:sigma-y}  & \sigma_i^{y} = -i \sigma^z_{1} \cdots \sigma^z_{i-1} J( \F[i] - \F[i]^\dagger ), \\
  \label{eq:sigma-z}  & \sigma_i^{z} = J(\F[i]^\dagger \F[i] - \F[i] \F[i]^\dagger ).
\end{align}

Notice that, the parity superselection rule in the Fock space $\Hilb{F}_N$ is trivially translated in the
qubit space thanks to the \Jwt; \ie defining \emph{total occupation number} for the qubit vector \(
\Ket{s_1,\ldots,s_N}_\mathrm{Q} \) as the sum \( \sum_i s_i \), the Wigner superselection forbids states that
are projections on superpositions of qubit vectors with total occupation numbers different modulo $2$.

\subsection{Bilocal tomography}

In the following, we exploit the \Jwt to represent the states of the \Fqt. For the sake of simplicity, we will
drop the $J$ symbol when this causes no confusion.

Thanks to Section~\ref{s:fqt}, we know that \Fqt is the parity superselected version of the \Qt of qubits.
Using the generalized theory of superselected \Opt developed in Section~\ref{s:superselected-opt} we can see
that \Fqt can be regarded as a minimal superselection of \Qt:

\begin{proposition}\label{prop:minimal-superselection}
  \Fqt is a minimal superselection of \Qt with the following linear constraints on the qubit system \( \sB \)
  \begin{equation}\label{eq:Fermionic-linear-constraints}
    \SetStates(\sigma(\sB))\coloneqq\{\rho\in\SetStates(\sB);\, \Tr[\sigma^x \rho]=\Tr[\sigma^y \rho] =0\}.
  \end{equation}
\end{proposition}

\begin{proof}
  Let $\sA$ and $\sB$ be a 1-\Lfm and a 1-qubit system, respectively. We have noticed already that, due to the
  parity prescription, \( \sA \) has only two pure states $\KetBra{0},\KetBra{1}$. Then the density matrices
  $\rho\in\SetStates(\sA)$ shall be diagonal
  \begin{equation}\label{eq:local-constraints}
    \Tr[\sigma^x \rho]=\Tr[\sigma^y \rho] =0, 
  \end{equation}
  showing that the superselection on the elementary systems is as in
  Eq.~\eqref{eq:Fermionic-linear-constraints} with
  \begin{equation*}
    \sA = \sigma(\sB)\qquad\qquad D_{\sA} = D_{\sB}-V_{\sB}^{\sigma},\qquad V_{\sB}^{\sigma}=2.
  \end{equation*}

  Now we have to show that the whole \Fqt is built bottom-up extending in the minimal way the constraint
  \eqref{eq:Fermionic-linear-constraints} on the composite systems. Let $\sB_N$, $\sB_M$ be two systems made
  of $N$ and $M$ qubits respectively.  According to definition \ref{def:minimal-superselection} we can simply
  check that \Fqt achieves the lower bound of Eq.~\eqref{eq:constraints-on-the-constraints-a},
  \begin{equation}\label{eq:Fermionic-minimal-check}
    V_{\sB_N \sB_M}^{\sigma} = D_{\sB_N} V_{\sB_M}^{\sigma} + D_{\sB_M} V_{\sB_N}^{\sigma} - 2
       V_{\sB_N}^{\sigma} V_{\sB_M}^{\sigma}.
  \end{equation}
  Using Eq.~\eqref{eq:Fermionic-dimension} we have \( V_{\sB_N}^\sigma = \frac{1}{2} D_{\sB_N} \) and  \(
  V_{\sB_M}^\sigma = \frac{1}{2} D_{\sB_M} \), hence Eq.~\eqref{eq:Fermionic-minimal-check} is satisfied.
\end{proof}

Since \Qt is local-tomographic, thanks to Prop.~\ref{prop:minimal-maximal} the \Fqt is maximally
bilocal-tomographic. This can also be verified counting the number of independent local and 2-local effects
for a system of $N$ {\Lfm}s and noticing that it is exactly its states space dimension:
\begin{equation*}
  \sum_{k=0}^{\lfloor N/2\rfloor} \binom {N} {2k}D_{1\,\Lfm}^{N-2k}\tilde
    D_{2\,\Lfm\textrm{s}}^k=2^{2N-1}=D_{N\,\Lfm\textrm{s}}.  
\end{equation*}

We emphasize that \Fqt provides an example of a bilocal-tomographic theory whose systems do not satisfy the
dimensional prescription in Ref.~\citen{hardy2012limited}. Indeed, after showing that the dimension of the
non-local component of a bipartite system $\tilde D_{\sA\sB}=D_{\sA\sB}-D_{\sA}D_{\sB}$ can be factorized as
$\tilde D_{\sA\sB}=L_{\sA}L_{\sB}$, and assuming that the two functions
\begin{equation*}
  D_{\sA}+L_{\sA},\qquad  D_{\sA}-L_{\sA},
\end{equation*}
are strictly increasing functions of the number of perfectly discriminable states $d_{\sA}$, the authors of
Ref.\citen{hardy2012limited} prove that in a bilocal-tomographic theory, the dimension of the system $\sA$
must be
\begin{equation}\label{eq:HW-bilocal-dimension}
  D_\sA=\frac{1}{2}(d_{\sA}^r+d_{\sA}^s),
\end{equation}
for some integers $r,s$ satisfying $r\geq s>0$. This is not true for the Fermionic computation where for
example $D_{\sA}=8$ cannot be achieved in this way. The strict monotonicity of the function $D_{\sA}-L_{\sA}$
is too restrictive and excludes the Fermionic case from the set of admissible bilocal-tomographic theories,
since we have $D_{\sA}-L_{\sA}=0$ for any system $\sA$ made of an arbitrary number $N$ of {\Lfm}s.

Not satisfying Local Tomography, the \Fqt does not satisfy the property of Remark~\ref{rem:local-tomography}
in Section~\ref{sub:lobi}. Indeed consider the unitary maps on a single \Lfm system given by $ I $, $ \sigma^x
$, $ \sigma^y $, $ \sigma^z $. Being \( \Ket{0} \), \( \Ket{1} \) the only pure normalized states of a single
\Lfm, the maps $ \sigma^x $ and \( \sigma^y \) (and similarly \( I \) and \( \sigma^z \)) are equal when
evaluated on a one-\Lfm system $\sA$; pictorially:
\begin{equation*}
  \begin{aligned}
    \Qcircuit @C=1em @R=.7em @! R { 
      \prepareC{\rho} & \ustickcool{\sA} \qw & \gate{\sigma^x} & \ustickcool{\sA}\qw & \qw
    }
  \end{aligned} 
  =
  \begin{aligned}
    \Qcircuit @C=1em @R=.7em @! R { 
      \prepareC{\rho} & \ustickcool{\sA} \qw & \gate{\sigma^y} & \ustickcool{\sA}\qw & \qw
    }
  \end{aligned},
  \qquad \forall \State{\rho} \in \SetStates(\sA).
\end{equation*}
We need a $n$-\Lfm state with $n\geq0$ to verify that the two maps \( \sigma^x \) and \( \sigma^y \) are
indeed different; \eg considering \( \Ket{\Psi} = \alpha \Ket{00} + \beta \Ket{11} \) we get \( \sigma^x
\otimes I \Ket{\Psi} \ne e^{i \gamma} \sigma^y \otimes I \Ket{\Psi} \) for every \( \gamma \in \Reals \);
pictorially
\begin{equation*}
  \begin{aligned}
    \Qcircuit @C=1em @R=.7em @! R { 
      \multiprepareC{1}{\Psi} & \ustickcool{\sA} \qw & \gate{\sigma^x} & \ustickcool{\sA}\qw
        & \qw \\
      \pureghost{\Psi} & \qw & \qw & \ustickcool{\sB} \qw 
        & \qw 
    }
  \end{aligned} 
  \ne
  \begin{aligned}
    \Qcircuit @C=1em @R=.7em @! R { 
      \multiprepareC{1}{\Psi} & \ustickcool{\sA} \qw & \gate{\sigma^y} & \ustickcool{\sA}\qw
        & \qw \\
      \pureghost{\Psi} & \qw & \qw & \ustickcool{\sB} \qw 
        & \qw 
    }
  \end{aligned}.
\end{equation*}

\subsubsection{Other superselected quantum theories}\label{s:other-superselected-theories}

It is worth mentioning that \Fqt is not the unique minimal superselection of \Qt.  Another example is given by
the {\em Real Quantum Theory} (\Rqt) defined \cite{hardy2012limited} as the restriction of the quantum case to
real matrices. The elementary system of \Rqt, with two perfectly distinguishable states, is denoted {\em
rebit} and its convex set of states is the disk obtained by the equatorial section of the qubit.  According to
Definition~\ref{def:superselection-rule}, the \Rqt is a superselection of the standard \Qt, being the
requirement of reality of a quantum state $\rho$ given by the linear constraint $\rho-\rho^T=0$, with $T$
denoting transposition with respect to a fixed basis taken as real.  Hence, if $\sA$ is the multipartite
system of $N$ \emph{rebits} having Hilbert dimension $d_{\sA}=2^N$, the dimension of
$\SetStates_{\Reals}(\sA)$ is given by
\begin{equation*}
  D_{\sA} = d_{\sA} ( d_{\sA} + 1 ) / 2.
\end{equation*}
Thus, if $\sB$ is a system of $N$ qubits, one has $ \sA = \sigma(\sB)$ where the number of linear constraints
for the system $\sA$ is given by
\begin{equation*}
  V_{\sA}^{\sigma}=D _{\sB}-D _{\sA}
\end{equation*}
One can easily check that also \Rqt is minimally superselected; indeed, for a couple of systems $\sA$, $\sC$ of
$N$ and $M$ rebits respectively, the number of constraints for the composite system $\sA\sC$
\begin{equation*}
  V_{\sA\sB}^{\sigma}=\tfrac{1}{2} d_{\sA}d_{\sB}(d_{\sA}d_{\sB}-1)
\end{equation*}
saturates the lower bound of Eq.~\eqref{eq:constraints-on-the-constraints-b}. Hence, from the linear
constraint of a 1-rebit system \( \Tr[\sigma^y\rho] = 0 \), we build the whole \Rqt by taking the minimal
extension of this constraint to the composite systems.  Therefore, according to
Proposition~\ref{prop:minimal-maximal}, the \Rqt is maximally bilocal-tomographic (see also
Ref.~\citen{hardy2012limited}).

In Proposition \ref{prop:minimal-maximal} we have considered the extremal cases of minimal and maximal
superselection, which lead respectively to bilocal- and local-tomographic theories. On the other hand, there
is a full range of possible constraints between these two cases---\ie $V^\sigma_{\sA\sB}$ strictly included in
the bounds of Eqs.~\eqref{eq:constraints-on-the-constraints-a} and
\eqref{eq:constraints-on-the-constraints-b}--- where one can find superselected theories with different
degrees of discriminability. 

As already pointed out at the end of Section~\ref{sec:JW}, the parity superselection of the \Fqt is trivially
translated in the \Qt representation by allowing only pure qubits states that are projections on
superpositions of vectors with total occupation numbers equal modulo $2$. A more general scenario is given by
considering a \emph{number superselected \Qt}, namely superselected \Qt theories of qubits where the
admissible pures states are projections on superpositions of vectors with the same total occupation numbers.

\begin{proposition}
  There is no $n\in\mathbb{N}$ such that a number superselected \Qt is $n$-local tomographic.
\end{proposition}
\begin{proof}
  For any $n$ we will present a suitable composite system \( \sB \coloneqq \sB_1\ldots\sB_n\sB_{n+1} \) and a
  couple of state \( \State{\psi_+}, \State{\psi_-} \in \SetStates(\sB) \) we cannot distinguish by means of
  $n$-local effects (see Definition~\ref{def:n-local-tomography}). Set $n$ be an arbitrary integer, then for
  $1 \leq i \leq n$ each subsystem $\sB_i$ is the elementary system of the number superselected \Qt, while
  $\sB_{n+1}$ is the parallel composition of $n$ of such elementary systems.  $\State{\psi_+}, \State{\psi_-}$
  are the pure states corresponding to the projections on the Hilbert space vectors
  \begin{equation*}
    \Ket{\psi_\pm} \coloneqq \Ket{0}_{\sB_1} \Ket{0}_{\sB_n} \Ket{1,1,\ldots,1}_{\sB_{n+1}} \pm
      \Ket{1}_{\sB_1} \Ket{1}_{\sB_n} \Ket{0,0,\ldots,0}_{\sB_{n+1}}.
  \end{equation*}
  There is no $n$-local effect able to discriminate the two states, \ie no discriminating effect has the form
  $E^{(n)} \otimes E^{(1)}$, with $E^{(n)}$ an effect for $n$ subsystems (hence either $E^{(n)} \in
  \Herm((\Complexes^2)^{\otimes n})$ or $E^{(n)} \in \Herm((\Complexes^2)^{\otimes (n-1)} \otimes
  (\Complexes^2)^{\otimes n})$), and $E^{(1)}$ a $1$-effect (and so either $E^{(1)} \in
  \Herm((\Complexes^2)^{\otimes n }$ or $E^{(1)} \in \Herm(\Complexes^2$). Indeed, since the two states differ
  only in the sign of the off-diagonal terms, a suitable effect to tell them apart should have a non null
  component \( \KetBra{\psi_+}[\psi_-] = \KetBra{0}[1]_{\sB_1} \otimes \ldots \otimes \KetBra{0}[1]_{\sB_n}
  \otimes \KetBra{1\ldots1}[0\ldots0]_{\sB_{n+1}} \). However $ \KetBra{\psi_+}[\psi_-] $ cannot be spanned by
  the tensor product of the  two effects \( E^{(1)} \), \( E^{(n)} \), since due to the superselection rule
  each \( E^{(1)} \), \( E^{(n)} \) have just the matrix elements \( \KetBra{s_1,\ldots,s_k}[t_1,\ldots,t_k]
  \) with \( \sum s_i = \sum t_i\).
\end{proof}

The previous result shows that a \Qt with number superselection has a cumbersome tomographic property: given a
$n$-partite system there is always a couple of states that cannot be discriminated without resorting to a
non-local effect involving all the subsystems.

\subsection{Fermionic entanglement}

Entanglement is commonly regarded as the peculiar trait of \Qt and it has been studied extensively also in
relation to the other quantum features. A pure state of a pair of quantum systems is called entangled if it
cannot be factorized, while a mixed state is entangled if it cannot be written as a mixture of factorized pure
states, \ie it is not {\em separable}. The main goal in the study of entanglement is to find criteria for
testing whether a state is separable or not (see for example the {\em partial transpose condition} proposed by
Peres in Ref.~\citen{PhysRevLett.77.1413}), and to provide consistent measures for quantifying entanglement.
Among the measures of entanglement considered in the literature we can cite the {\em entanglement of
formation} \cite{bennett1996mixed, hill1997entanglement, PhysRevLett.80.2245,wootters2001entanglement}, the
{\em distillable entanglement} \cite{PhysRevA.51.1015}, and the {\em relative entropy of entanglement}
\cite{vedral1997quantifying} (for a review on the entanglement measures see
Ref.~\citen{plenio2005introduction}).

Despite entanglement in \Qt has been largely investigated, the nature of entanglement in general {\Opt}s is
almost an unexplored field. Because of the physical relevance of the Fermionic field some
authors\cite{PhysRevA.76.022311} have recently wondered how separable states can be defined for Fermionic
systems, taking into account the non-local action of mode creation and annihilation operators. Here we study
the entanglement in \Fqt and show how the parity superselection derived in
Section~\ref{s:parity-superselection} affects the features of the resulting theory. 

While the notion of entangled state as a non-separable state---\ie a state that cannot be prepared by
\Locc---can immediately be generalized to arbitrary {\Opt}s, it is not clear whether this notion is
operationally relevant in the absence of local tomography or not. For example, it may be that in order to
discriminate an entangled state from a separable one, one needs bipartite effects, and then one cannot use
this kind of entanglement to violate Bell-like inequalities. The non triviality of the operational notion of
Fermionic entanglement has been the focus of Ref.~\citen{PhysRevA.76.022311}. There the authors propose four
different definitions of entanglement for Fermionic systems and provide a careful analysis of their mutual
relations. Fortunately, as we will see in this section, it turns out that in \Fqt any entangled state can be
discriminated from any separable one by local effects, provided that two copies of the state are available,
thus establishing non-separability as the unique notion of entanglement in \Fqt.

Once an \Opt is provided with a notion of ``entangled state'', the amount of entanglement in a given state of
the theory should be quantified in operational terms. Having the notion of \emph{entanglement of formation} a
clear operational interpretation, here we will extend this measure of quantum entanglement to the Fermionic
case.  The entanglement of formation, introduced in Ref.~\citen{bennett1996mixed} and in
Ref.~\citen{hill1997entanglement}, focuses on the resources needed in order to generate a given amount of
entanglement when state manipulation is restricted to \Locc. In \Qt all measures of entanglement for bipartite
states refer to a standard unit: the {\em ebit}, which is the amount of entanglement of a bipartite
\emph{singlet state}. The entanglement of formation of a quantum state $\rho$ represents the minimum number of
ebits needed to achieve a decomposition of $\rho$ into pure states by means of \Locc, where the minimization
is over all possible decompositions. The constraint of \Locc plays a fundamental role in order to view
entanglement as a \emph{resource}.  Indeed, the amount of entanglement does not increase under \Locc
transformations, inducing a hierarchy of states based on their ``usefulness'' under \Locc operations.
Accordingly a state is called \emph{maximally entangled} when it can be transformed into any other by means of
\Locc. In \Qt we can find a single two-qubit state that can be used to achieve all the other two-qubit states
by means of \Locc: the singlet state. As soon as we increase the dimension of quantum systems, it is no longer
possible to identify a unique state we can use to get all the others\cite{Kraus}.  The customary notion of
maximally entangled state has to be superseded by that of \emph{maximally entangled set} (\Mes) of $n$-partite
states, namely the set of states maximally useful under \Locc manipulation, \ie any state outside this set can
be obtained via \Locc from one of the states within the set, and no state in the set can be achieved from any
other state via \Locc. It is still not clear in \Qt whether the \Mes is stable once we study the asymptotic
quantification of entanglement. In a general \Opt we cannot expect that the \Mes for bipartite states reduces
to a unique maximally entangled state, as in \Qt for bipartite qubit entanglement.

Allowing for classical communication in \Locc implies that \Locc protocols are not completely local,
introducing a complicate structure whose complete characterization is still an open problem in \Qt. A full
theory of Fermionic entanglement would require the introduction of similar notions, involving a complete
analysis of the transformations of states under \Locc, which goes beyond the scope of this paper.
Nevertheless we can find some relevant features of entanglement in \Fqt, and show that \Fqt and \Qt are very
different from the entanglement point of view.

Here is a brief summary of the results presented in this section:
\begin{enumerate}[a.]
  \setlength{\itemsep}{1pt} \setlength{\parskip}{2pt}
  \setlength{\parsep}{1pt}
  \item  non-separability is the unique notion of entanglement in \Fqt;
  \item there is a simple linear criterion for testing the full separability of states;
  \item Fermionic \Locc correspond to quantum \Locc with a polynomial overhead of classical communication;
  \item \Mes are needed also for bipartite states; 
  \item there are mixed states that are not separable and with maximal entanglement of formation; 
  \item there are states with maximal entanglement of formation that do not belong to a \Mes; 
  \item the monogamy of entanglement is violated (taking as measure the Fermionic concurrence in relation with
    the Fermionic entanglement of formation).
\end{enumerate}
Some of these results can also be found in Ref.~\citen{d2013Fermionic}. 

Again, in the following we exploit the \Jwt to represent the states of the \Fqt and we will drop the $J$
symbol for the sake of clarity.

\subsubsection{Non-separability as the unique notion of Fermionic entanglement}

We show that in \Fqt any entangled state can be discriminated from any separable one by local effects,
provided that two copies of the state are available. This feature stems from Theorem \ref{th:jellyfish} and
indicates non-separability as the unique notion of entanglement in \Fqt. 

Suppose that two states $\rho$ and $\sigma$ in $\SetStates(\sA\sB)$ are different. This implies that there
exists an effect $a\in\SetEffects(\sA\sB)$ such that $\RBraKet{a}{\rho}\neq\RBraKet{a}{\sigma}$. Either $a$ is
in $\SetEffects_{\mathbb R}(\sA)\otimes\SetEffects_{\mathbb R}(\sB)$, in which case local measurements are
sufficient to discriminate between $\rho$ and $\sigma$, or $a$ has a genuinely bipartite term in
$\widetilde\SetEffects_{\mathbb R}({\sA\sB})\coloneqq\SetEffects_{\mathbb R}(\sA\sB) /\SetEffects_{\mathbb
R}(\sA)\otimes\SetEffects_{\mathbb R}(\sB)$, where the quotient is modulo the equivalence relation $a\sim b$
iff $a-b\in\SetEffects_{\mathbb R}(\sA)\otimes\SetEffects_{\mathbb R}(\sB)$. This implies that if we have to
discriminate between $\rho\otimes\rho$ and $\sigma\otimes\sigma$, we need an effect in
$\widetilde\SetEffects_{\mathbb R}({\sA\sB})\otimes\widetilde\SetEffects_{\mathbb R}({\sA'\sB'})$.  Now, by
theorem \ref{th:jellyfish}, this space is also spanned by functionals in $\widetilde\SetEffects_{\mathbb
R}({\sA\sA'})\otimes\widetilde\SetEffects_{\mathbb R}({\sB\sB'})$.  Finally, this means that a factorized
effect $c\otimes d$ with $c\in\SetEffects(\sA\sA')$ and $d\in\SetEffects(\sB\sB')$ is sufficient to detect
entanglement between Alice's systems $\sA\sA'$ and Bob's $\sB\sB'$. Any state that is not separable is then
actually entangled in any operational sense, namely its statistics on \Locc effects is different from that of
any separable state.  Notice also that two copies of the state are sufficient to detect entanglement.

\subsubsection{Full Separability criterion for multi-\Lfm states}\label{s:separability-criterion}

Unlike \Qt, \Fqt admits a linear criterion for establishing whether a state of many {\Lfm}s is fully
separable.  By definition, a state of $N$ {\Lfm}s is fully separable if it can be written as a convex
combination of product states, namely
\begin{align}\label{eq:full-separability}
  \rho=\sum_{i=1}p_i \rho^{(i)}_1\otimes\rho^{(i)}_2\otimes\cdots\otimes\rho^{(i)}_N, \qquad 
    \text{with } \sum_i p_i=1,\ p_i\geq 0. 
\end{align}
Since the local states of the $i$-th \Lfm are convex combination of $\KetBra{0}_i$ and $\KetBra{1}_i$, an
arbitrary $N$-{\Lfm}s state is fully separable if and only if it is diagonal in the Fock basis of vectors
$\ket{s_1,\ldots,s_N}$.  If we consider now a state of a system $\sA$ made of $N$ composite systems
$\sA_1,\sA_2,\dots,\sA_N$, by definition a state of $\sA$ is separable if it can be expressed as in
Eq.~\eqref{eq:full-separability}, with $\rho^{(i)}_j\in\SetStates(\sA_j)$. Then it is clear that a necessary
condition for separability is that the full state $\rho$ commutes with all local parity operators. Moreover, a
state $\rho\in\SetStates(\sA)$ that commutes with local parity operators is separable if and only if the
projections of $\rho$ in every parity sector correspond to density matrices of quantum separable states.

\subsubsection{Fermionic \Locc}\label{sec:locc}
We will now show that every \Locc protocol in the \Fqt is simulated by a \Locc protocol in \Qt. Notice that we
can find three classes of \Fqt transformations: (i) transformations whose Kraus operators are even (\ie
superpositions of products of even number of field operators), (ii) transformations whose Kraus operators are
odd, and (iii) transformations with both even and odd Kraus operators.  We can then refine every
transformation $\Transformation{T}$ to a test $\{\Transformation{T}_e,\Transformation{T}_o\}$ where
$\Transformation{T}_e$ has only even Kraus operators, while $\Transformation{T}_o$ has only odd ones. Thanks
to this decomposition we can prove the following lemma.

\begin{proposition}\label{p:Fermionic-locc} 
  Every Fermionic \Locc corresponds to a quantum \Locc on qubits under \Jwt.
\end{proposition} 
\begin{proof} 
  Let $\sC=\sC_1\ldots\sC_N$ be the Fermionic system made of $N$ {\Lfm}s, and let $\sA$ be one subsystems
  $\sA\coloneqq\sC_{i_1}\ldots\sC_{i_M}$ with \( i_j \in \chi_\sA \subseteq \Gamma_N \coloneqq \{ 1, \ldots, N
  \} \) made of $M$ {\Lfm}s.  Consider now the most general bipartite \Locc on $\sC$ between Alice controlling
  the subsystem $\sA$ and Bob controlling the subsystem $\sB$ complementary to $\sA$ (\ie $\sC = \sA\sB$). One
  can always sort the {\Lfm}s in the Jordan-Wigner representation so that the first $N-M$ {\Lfm}s correspond
  to Bob's subsystem $\sB$. Denoting with $E_\sX$ and $O_\sX$ an even and an odd Kraus operator for the
  subsystem $\sX$, the \Jwt maps single-Kraus transformations local on $\sA$ and $\sB$ in the following way 
  \begin{align}
    & J(O_A)=\bigotimes_{i\in\Gamma_N \setminus \chi_\sA}\!\!\!\!\!\! \sigma_i^z\ \otimes O^\prime_A,\qquad & 
      J(E_A)=I_B\otimes E^\prime_A, \label{eq:Fermionic-quantum-locc1} \\
    & J(O_B)=O^\prime_B \otimes I_A,\qquad & J(E_B)=E^\prime_B\otimes I_A, \label{eq:Fermionic-quantum-locc2} 
  \end{align} 
  where $O^\prime_{\sX}$, $E^\prime_{\sX}$ correspond to Kraus operators of quantum maps on the subsystem
  $\sX$.  Equations \eqref{eq:Fermionic-quantum-locc1}-\eqref{eq:Fermionic-quantum-locc2} show that if Alice
  and Bob perform a Fermionic \Locc protocol, this is equivalent to a quantum \Locc protocol in the \Jwt
  representation: Indeed, whenever Alice needs to apply a Fermionic transformation $\Transformation{T}$, she
  can achieve it in the qubit case by performing the test $\{\Transformation{T}_e,\Transformation{T}_o\}$, and
  then she just needs to tell Bob whether the event $o$ or $e$ occurred: in this way Bob knows if he has to
  apply a string of $\sigma_z$ operators locally on his subsystem or not. On the other hand, Bob's Fermionic
  transformations are local also in the qubit case. We conclude that the \Jwt mapping preserves the \Locc
  nature of bipartite transformations, with an overhead of one classical bit at each round in order to
  communicate the parity of the Kraus operators.
  
  In general, one can consider an $n$-partite \Locc. In this case let $\sA_1, \ldots,\sA_n$ be the $n$
  subsystems partitioning $\sC$, and let us sort the {\Lfm}s such that the ones belonging to the system
  $\sA_i$ precede the ones of $\sA_j$ if $i<j$. The $i$th party needs a bit of classical information to
  communicate to the $(i-1)$th party the total parity of all the Kraus operators occurred up to the
  $(n-i+1)$th round, in this way the $(i-1)$th party knows whether he needs to apply the string of
  $\sigma_z$'s on his subsystems or not.  Iterating this process we find that a Fermionic $n$-partite \Locc
  corresponds under \Jwt to an $n$-partite qubit \Locc with an overhead of $n-1$ bits of classical
  information.
\end{proof}

\subsubsection{Maximally entangled sets for two {\Lfm}s}\label{ss:Fermionic-mes}

As already stated, the concept of ``maximally entangled state'' has to be superseded by that of \Mes
\cite{Kraus} even for two {\Lfm}s.  In \Fqt, a single \Lfm $\rho$ is operationally equivalent to a bit, so we
can perform locally only the unitary gates with Kraus
\begin{equation*}
  \sigma^x, \qquad \sigma^y, \qquad  \sin \vartheta\ I + i \sin \vartheta\ \sigma^z, \quad
    \vartheta\in\left[0,2\pi\right),
\end{equation*}
which do not allow to transform the vectors $\ket{0}$, $\ket{1}$ into any superposition. Thus, given a state
$\KetBra{\Psi}$ with Schmidth decomposition \( \Ket{\Psi}=\alpha\ket{00}+\beta\ket{11} \), one cannot change
the magnitude of the coefficients $\alpha$ and $\beta$ by local unitary operations. By acting locally one can
simply change the parity sector by means of the Kraus \( \sigma^x \), \( \sigma^y \) (which locally are the
same), and apply an arbitrary relative phase $\exp(2 i \vartheta)$ via the Kraus \( \sin \vartheta\ I + i \sin
\vartheta\ \sigma^z \).

We can moreover get any arbitrary factorized state of the \Fqt---\ie projections on \( \Ket{00} \), \(
\Ket{01} \), \( \Ket{10} \), \( \Ket{11} \)---from any state in the \Mes by means of \Locc operations: Alice
measures her \Lfm in the computational basis by the Kraus operators $\{\KetBra{0},\KetBra{1}\}$ and
conditionally on the outcome she tells Bob the local operation he has to apply on his \Lfm---\ie the identity
with Kraus $I$ or the bit flip with Kraus $\sigma^x$.  Clearly, one cannot do the opposite. Hence, examples of
\Mes's for two {\Lfm}s are given by $MES_0$, $MES_1$, which are defined as
\begin{align*}
  & MES_0 \coloneqq \{\KetBra{\Psi_{\alpha,\beta}} \mid
    \ket{\Psi_{\alpha,\beta}}:=\alpha\ket{00}+\beta\ket{11},\  \alpha,\beta > 0 \}, \\
  & MES_1 \coloneqq \{ \KetBra{\Psi_{\alpha,\beta}} \mid
    \ket{\Psi_{\alpha,\beta}}:=\alpha\ket{01}+\beta\ket{10}, \  \alpha,\beta > 0 \}.
\end{align*}

\subsubsection{The Fermionic entanglement of formation}

In the usual quantum theory scenario the \emph{entanglement cost} of a given, generally entangled, state
$\rho\in\SetStates(\sA\sB)$ shared by distant observers Alice and Bob quantifies the amount of resources
needed by the two parties in order to create the state $\rho$. Consider then the protocol
\begin{equation*}
  \KetBra{\Sigma}^{\otimes m}\xrightarrow{\Locc} \rho^{\otimes n}
\end{equation*}
where $m$ singlet states $\KetBra{\Sigma}$ are converted into $n$ copies of the target state $\rho$ by means
of \Locc. Perfect transformation by \Locc is usually impossible and one requires it only asymptotically, say
in the limit where the number of created copies of $\rho$ approaches infinity.  The entanglement cost $E_c$ is
thus defined as the optimal asymptotic ratio $r=m/n$. The last one is very difficult to compute, while the
\emph{entanglement of formation}, which also has an operational interpretation, can be more easily computed in
terms of the density matrix $\rho$.

The definition of entanglement of formation is based on the result of Ref.~\citen{PhysRevA.56.R3319} for the
entanglement cost of pure states. In the paper the authors show that the entanglement cost of a pure state
$\rho=\KetBra{\Psi}[\Psi]$ coincides with the von Neumann entropy of either of its marginal states, say
$E_c(\KetBra{\Psi})=S(\Tr_\sA\KetBra{\Psi})$. Therefore to produce $\KetBra{\Psi}^{\otimes n}$ one needs $m
\approx n S(\Tr_\sA\KetBra{\Psi})$ singlets with the equality achieved in the asymptotic limit.  The
entanglement of formation of a mixed state $\rho\in\SetStates(\sA\sB) $ is then defined as
\begin{equation}\label{eq:entanglement-of-formation-qt}
  E(\rho) \coloneqq \min_{\mathcal{D}_\rho} \sum_i p_i S(\Tr_\sA\KetBra{\Psi_i}),
\end{equation}
where
\begin{equation*} 
  {\mathcal D}_\rho\coloneqq \{\{ p_i, \Ket{\Psi_i} \}\mid \rho = \sum_i p_i \KetBra{\Psi_i} \}
\end{equation*}
is the set of all the pure decompositions of the mixed state $ \rho $. The operational
interpretation\footnote{Notice that the
  entanglement of formation of a mixed state $\rho$ is not proven to correspond to its entanglement cost, and
  in general it is $E(\rho)\geq E_c(\rho)$. However, in Ref.~\citen{hayden2001asymptotic} it has been shown
  that
  \begin{equation*}
    E_c(\rho)=\lim_{n\to\infty} E(\rho^{\otimes n})/n,
  \end{equation*}
  where the right hand side of the equality is the so called \emph{regularized entanglement of formation}.
  If the entanglement of formation turns out to be additive, the entanglement cost will be equal to the
  entanglement of formation.} 
of the entanglement of formation has been pointed out by Wootters in Ref.~\citen{wootters2001entanglement},
where it is noticed that
\begin{equation}\label{eq:quantum-entanglement-of-formation-op}
  E(\rho) \equiv \lim_{n\to\infty} m_n(\rho)/ n,
\end{equation}
with $m_n(\rho)$ the minimum number of singlet states needed by two parties to prepare via \Locc random tensor
products $\bigotimes_{l=1}^n\KetBra{\Psi_{i_l}}$ of states in a decomposition $\{p_i,\KetBra{\Psi_i}\}$ of
$\rho$, sampled by the distribution $p(i_1,\ldots,i_n)=p(i_1)\ldots p(i_n)$, minimized over all possible
decompositions:
\begin{equation*}
  \KetBra{\Sigma}^{\otimes m}\xrightarrow{\Locc,\,{\mathcal D}_\rho} \rho^{\otimes n}.
\end{equation*}

In the Letters \citen{hill1997entanglement, PhysRevLett.80.2245} it is also provided a formula for evaluating
the entanglement of formation \eqref{eq:entanglement-of-formation-qt} of a state $\rho$ just in terms of its
density matrix. For a mixed state $ \rho$ of two qubits one has
\begin{equation}\label{eq:entanglement-of-formation-concurrence}
  E(\rho)=\mathcal{E}(C(\rho))  
\end{equation}
with $ \mathcal{E}(x) \coloneqq h(\tfrac{1+\sqrt{1-x^2}}{2}) $, $ h $ the binary Shannon entropy, and the
expression of the \emph{concurrence} $C(\rho) $ depending only on the density matrix $\rho$ (see
Refs.~\citen{hill1997entanglement, PhysRevLett.80.2245} for the explicit formula of the concurrence). As for
the entanglement of formation, also the concurrence of a generally mixed state $\rho$ is given by
\begin{equation}\label{eq:concurrence-qt}
  C(\rho) \coloneqq \min_{\mathcal{D}_\rho} \sum_i p_i C(\KetBra{\Psi_i}).
\end{equation}
Both the entanglement of formation and the concurrence are zero if and only if the state $ \rho $ is
separable, and for two qubits they reach the maximum value $ 1 $ if and only if $ \rho $ is a maximally
entangled state.

In analogy to the quantum case we can define the operational Fermionic entanglement of formation. Given a
Fermionic state $\rho=p_0\rho_0+p_1\rho_1$ its entanglement of formation is defined as
\begin{align}\label{eq:Fermionic-entanglement-of-formation-op}
  E_{\mathrm{F}}(\rho)=\lim_{n\to\infty} m_n(\rho)/n,
\end{align}
with $ m_n(\rho)$ the minimum number of states in a Fermionic \Mes needed by two parties to prepare via
Fermionic \Locc random tensor products $\bigotimes_{l=1}^n\KetBra{\Psi_{i_l}}$ of states in a decomposition
$\{p_i,\KetBra{\Psi_i}\}$ of $\rho$, sampled by the distribution $p(i_1,\ldots,i_n)=p_{i_1}\ldots p_{i_n}$,
minimized over all possible decompositions:
\begin{align}\label{eq:prot0}
  \bigotimes_{k=1}^{m} \KetBra{\Sigma_k} \quad \left(\KetBra{\Sigma_k}\in\Mes\right)
    \xrightarrow{\Locc_{\mathrm{F}},\,{\mathcal D}_{\rho}^{\mathrm{F}}}    \rho^{\otimes n}.
\end{align}
It is important to notice that the Fermionic entanglement of formation of a mixed state corresponds to the
convex-roof extension of the Fermionic entanglement of formation of pure states, as follows
\begin{equation}\label{eq:Fermionic-entanglement-of-formation-op-mixed}
  E_{\mathrm{F}} (\rho) = \min_{\mathcal{D}_\rho^{\mathrm{F}}} \sum_i p_i E_F(\KetBra{\Psi_i}),
\end{equation}
where $ \mathcal{D}_\rho^{\mathrm{F}} $ is the set of all the pure decompositions of $ \rho $ satisfying the
parity superselection rule\footnote{In Ref.~\citen{caves} the authors do the same for \Rqt considering the
  decompositions $\mathcal{D}_\rho^{\mathrm{R}}$ on real states.}. For pure states, we have the following
result.

\begin{proposition}\label{prop:Fermentoffpure}
  For pure states $\KetBra{\Psi}$, the function 
  \begin{equation}\label{eq:entanglement-of-formation-pure-fqt}
    \tilde E_F(\KetBra{\Psi}):=S(\Tr_A\KetBra{\Psi})
  \end{equation}
  is a lower bound for the Fermionic entanglement of formation
  \eqref{eq:Fermionic-entanglement-of-formation-op}.
\end{proposition}
\begin{proof}
  First notice that ${E}_{\mathrm{F}}(\KetBra{\Psi})$ in Eq.~\eqref{eq:Fermionic-entanglement-of-formation-op}
  corresponds to the maximal rate of conversion of states in the Fermionic \Mes to the state $\KetBra{\Psi}$
  via fermionic {\Locc}s, as in Eq.~\eqref{eq:prot0}.  Now consider the following protocol for qubit states
  \begin{align}\label{eq:prot2}
    \KetBra{\Sigma}^{\otimes m^\prime} 
       \xrightarrow{\Locc}
    \bigotimes_{k=1}^{m} \KetBra{\Sigma_k} \quad\left(\KetBra{\Sigma_k}\in\Mes\right)    
       \xrightarrow{\Locc_{\mathrm{F}}} \KetBra{\Psi}^{\otimes n}
  \end{align}
  where $m^\prime$ quantum singlets are converted via quantum \Locc into $m$ states in the Fermionic \Mes that
  are then converted at a rate $E^{\mathrm{op}}_{\mathrm{F}}(\KetBra{\Psi})$ into $n$ copies of the target
  state $\KetBra{\Psi}$ via Fermionic \Locc. Since any Fermionic state in the \Mes has a quantum entanglement
  of formation smaller than (or equal to) 1, the protocol \eqref{eq:prot2} allows for a conversion rate
  \begin{align}
    \frac{m^\prime}{m}{E}_{\mathrm{F}}(\KetBra{\Psi})\leq {E}_{\mathrm{F}}(\KetBra{\Psi}) .
  \end{align}
  Moreover, since any Fermionic \Locc is also a quantum \Locc (see Proposition~\ref{p:Fermionic-locc}) the
  protocol \eqref{eq:prot2} is a particular instance of the general protocol for \Locc conversion of $m'$
  singlet states to $n$ copies of the target state $\KetBra{\Psi}$, and then we have
  \begin{align}
    \tilde E_{\mathrm{F}}(\KetBra{\Psi}) = E(\KetBra{\Psi}) \leq \frac{m^\prime}{m}
      E_{\mathrm{F}}(\KetBra{\Psi}) \leq {E}_{\mathrm{F}}(\KetBra{\Psi}).
  \end{align}
  This proves the thesis.
\end{proof}

Now, if we extend the definition of $\tilde E_\mathrm F(\rho)$ to mixed states by convex-roof extension, we
have
\begin{align}
  \tilde{E}_{\mathrm{F}} (\rho) &\coloneqq \min_{\mathcal{D}_\rho^{\mathrm{F}}} \sum_i p_i \tilde E_\mathrm{F} 
    (\KetBra{\Psi_i})=p_0 E(\rho_0) + p_1 E(\rho_1), \label{eq:entanglement-of-formation-fqt} \\
  C_{\mathrm{F}} (\rho) &\coloneqq \min_{\mathcal{D}_\rho^{\mathrm{F}}} \sum_i p_i C(\KetBra{\Psi_i})=p_0
    C(\rho_0) + p_1 C(\rho_1), \label{eq:concurrence-fqt}
\end{align}
where we introduced the quantity $C_{\mathrm{F}} (\rho)$ that extends the notion of {\em concurrence} to the
Fermionic case\footnote{The expressions \eqref{eq:entanglement-of-formation-fqt} and
  Eq.~\eqref{eq:concurrence-fqt} were already proposed in Ref.~\citen{PhysRevA.76.022311}. Here we show that
  Eq.~\eqref{eq:entanglement-of-formation-fqt} provides a lower bound for the Fermionic entanglement of
  formation.}. The last equalities in Eqs.~\eqref{eq:entanglement-of-formation-fqt} and
\eqref{eq:concurrence-fqt} are obtained upon noticing that the state $\rho$ admits the unique
parity-decomposition $\rho = p_0 \rho_0 + p_1\rho_1$, with $p_0+p_1=1$ and $\rho_0$, $\rho_1$ states in the
even and odd parity sector respectively, and that all decompositions in $\mathcal{D}_\rho^{\mathrm{F}}$ shall
preserve the probabilities $p_0 $ and $p_1 $. Moreover, since $\mathcal{D}_{\rho_i}^{\mathrm{F}} \equiv
\mathcal{D}_{\rho_i} $, we have $ \tilde{E}_{\mathrm{F}} (\rho_i) = E(\rho_i) $ and $ C_{\mathrm{F}} (\rho_i)
= C(\rho_i) $.  Notice that for a pure state $\KetBra{\Psi}$ we have $C_F(\KetBra{\Psi})=C(\KetBra{\Psi})$.

Now, thanks to proposition \ref{prop:Fermentoffpure}, we clearly have
\begin{equation}
  \sum_ip_i\tilde E_\mathrm F(\KetBra{\Psi_i})\leq  \sum_ip_i E_\mathrm F(\KetBra{\Psi_i}),
\end{equation}
for every Fermionic pure state decomposition \( \{ p_i, \Ket{\Psi_i} \} \) of $\rho$, and then by
Eq.~\eqref{eq:Fermionic-entanglement-of-formation-op-mixed}
\begin{equation}
  \tilde E_\mathrm F(\rho)\leq E_\mathrm F(\rho).
\end{equation}

Notice that, unlike in \Qt \cite{hill1997entanglement} and in \Rqt \cite{caves}, the quantities $
E_{\mathrm{F}} $ and $ C_{\mathrm{F}} $ do not satisfy the relation $ E_{\mathrm{F}}(\rho) =
\mathcal{E}(C_{\mathrm{F}}(\rho)) $ (see Eq.~\eqref{eq:entanglement-of-formation-concurrence}). On the other
hand it is $\tilde E_{\mathrm{F}}(\rho) \geq \mathcal{E}(C_{\mathrm{F}}(\rho)) $, and for a state $\Phi$ with
$C_{\mathrm{F}}(\Phi)=1$ we have $\tilde E_{\mathrm{F}}(\Phi) = \mathcal{E}(C_{\mathrm{F}}(\Phi))=1 $.
Therefore, when $C_\mathrm F(\rho)=1$, the quantity $\tilde E_\mathrm F $ coincides with the operational
Fermionic entanglement of formation $E_\mathrm F$.

\subsubsection{Mixed states with maximal entanglement of formation} 

Using the quantities $ E_{\mathrm{F}} $ and $C_{\mathrm{F}} $, and the separability criterion we can show that
in \Fqt there are mixed states with maximal entanglement of formation. Consider the state
\begin{equation}\label{eq:entangled-mixed}
  \Phi\coloneqq\tfrac{1}{4}\left(I\otimes I+\sigma^x\otimes\sigma^x\right),
\end{equation}
corresponding to the mixture with $ p=1/2 $ of the Fermionic pure states $\KetBra{\Psi_0}$ and
$\KetBra{\Psi_1}$ with 
\begin{equation*}
  \Ket{\Psi_0} = \tfrac{1}{\sqrt{2}} \left( \Ket{00} + \Ket{11} \right), \qquad \Ket{\Psi_1} =
    \tfrac{1}{\sqrt{2}} \left( \Ket{01} + \Ket{10} \right).
\end{equation*}
Despite being mixed, $\Phi$ has maximal entanglement of formation and concurrence
\begin{align*}
  &  E_{\mathrm{F}}(\Phi) = \frac{1}{2} E( \KetBra{\Psi_0})+\frac{1}{2} E(\KetBra{\Psi_1})=1,\\
  &  C_{\mathrm{F}}(\Phi) = \frac{1}{2} C( \KetBra{\Psi_0})+\frac{1}{2} C(\KetBra{\Psi_1})=1.
\end{align*}
It is easy to verify that $\Phi$ is not separable; indeed $ \Phi $ does not satisfy the separability criterion
of Section~\ref{s:separability-criterion}
\begin{equation*}
  \Tr[\rho_\mathrm{sep}(\sigma^i\otimes\sigma^j)]=0,\qquad i,j=x,y.
\end{equation*}
Other mixed maximally entangled states can be found by replacing every occurrence of $\sigma^x_i$ with an
arbitrary linear combination of $\sigma^x_i$ and $\sigma^y_i$ in Eq.\eqref{eq:entangled-mixed}.  Notice that
all these states, which have maximal entanglement of formation, do not belong to a \Mes (see Section
\ref{ss:Fermionic-mes}).

Also \Rqt has mixed maximally entangled states. Being the rebit defined by the linear constraint $
\Tr[\sigma_y\rho]=0 $, in \Rqt a mixed maximally entangled state is achieved by replacing $ \sigma_x  $ with
$\sigma_y$ in the state of Eq.~\eqref{eq:entangled-mixed} \cite{caves}.

Notice that the state \eqref{eq:entangled-mixed} is separable in \Qt, since it is the mixture with $ p=1/2 $
of the pure product states $\KetBra{\Pi_+}$ and $\KetBra{\Pi_-}$ with
\begin{equation}
  \Ket{\Pi_+}:=\Ket{+}\Ket{+}, \quad  \Ket{\Pi_-}:=\Ket{-}\Ket{-},\qquad \text{and  }
    \Ket{\pm}=\tfrac{1}{\sqrt{2}}\left(\Ket{0}\pm\Ket{1}\right).
\end{equation}
Such a decomposition is not allowed neither in \Fqt nor in \Rqt, because of the violation of their respective
superselection rules by the vectors $ \Ket{\pm} $.

\subsubsection{Violation of entanglement monogamy}

The shareability of correlations between many parties is one of the main differences between quantum and
classical correlations.  While in the classical information theory correlations can be shared among arbitrary
many parties, in \Qt a system maximally entangled with a second system cannot share quantum correlations with
a third one.  This has been dubbed the ``monogamy of entanglement'' and a big effort has been devoted to its
quantification: see Refs.~\citen{PhysRevA.60.4344, PhysRevA.62.062314, PhysRevA.62.050302, PhysRevA.65.010301,
PhysRevA.61.052306, raggio1989quantum, fannes1988symmetric, terhal2004entanglement}, or
Refs.~\citen{PhysRevA.69.022309, RevModPhys.81.865} for a recent review on the subject.

Entanglement monogamy is usually stated by means of inequalities involving some entanglement measures, \ie
\begin{equation}\label{eq:monogamy}
  M(\rho_{\sA\sB}) + M(\rho_{\sA\sC}) \leq M(\rho_{\sA(\sB\sC)}),
\end{equation}
where $M(\rho_{\sA\sB})$ is a measure of the entanglement between systems $\sA$ and $\sB$. It is worth
mentioning that not every entanglement measure satisfies the inequality of Eq.~\eqref{eq:monogamy}, so not all
the entanglement measures are good indicators for monogamy. A measure satisfying the inequality
\eqref{eq:monogamy} is called {\em monogamous}. In \cite{PhysRevA.61.052306} it has been shown that in \Qt the
concurrence is monogamous and satisfies
\begin{equation}\label{eq:concurrence-monogamy}
  C^2(\rho_{\sA\sB}) +  C^2(\rho_{\sA\sC}) \leq 1.
\end{equation}
Notice that, if $\rho_{\sA\sB}$ has maximal concurrence---$C(\rho_{\sA\sB})=1$---then $\rho_{\sA\sC}$ must
have concurrence equal to $0$. 

The Fermionic entanglement (say the Fermionic concurrence) is not monogamous. For instance, consider the pure
state of three {\Lfm}s $\KetBra{\Phi'}$ with
\begin{equation}
  \Ket{ \Phi^\prime } \coloneqq \tfrac{1}{2}( \Ket{000} + \Ket{110} + \Ket{011} + \Ket{101} ).
\end{equation}
Tracing the state $ \KetBra{\Phi^\prime} $ over any one of the three {\Lfm}s, we find that the reduced
bipartite state is the mixed state $ \Phi $ of Eq.~\eqref{eq:entangled-mixed}, having maximal entanglement of
formation and concurrence. Therefore, in the \Fqt as well as in \Rqt \cite{wootters2012monogamy} each pair of
subsystems can share any amount of entanglement of formation.

\subsection{Fermionic computation}\label{sec:fermionic-computation}

Recently some authors have been wondering whether models of Fermionic quantum computation might support
universal computation and/or exhibit different computational power with respect to the standard quantum
computational model. As already stressed, one can build different computational models based on {\Lfm}s,
according to: (i) the degree of superselection on the states (\eg conservation of the parity number instead of
the total excitation number), and (ii) the admitted transformations of the theory. In
Ref.~\citen{Bravyi2002210} S.~B.~Bravyi and A.~Y.~Kitaev considered a \Lfm computational model with a parity
superselection where the unitary transformations are the parity-preserving ones (\ie the CP maps with a single
Kraus operator which is a linear combination of products of an even number of field operators). They showed
that such a computational model supports universal computation and that it can be simulated by regular unitary
gates of qubits with a computational overhead that goes as the logarithm of the number of the {\Lfm}s, thus
proving the computational equivalence of the two models.  The same result can be extended to the \Fqt
presented in this paper which is the largest computational model based on {\Lfm}s satisfying the assumptions
(\ref{a:causality})--(\ref{a:deteff}) in Section~\ref{sub:ass}.

In extending the results of Ref.~\citen{Bravyi2002210} to the \Fqt, we also review the original proofs for the
sake of completeness. The proofs for the \Fqt relies on the following observation: unitary transformations of
Fermionic quantum computation of Ref.~\citen{Bravyi2002210} are parity-preserving, while the \Fqt allows also
parity-changing transformations, \ie the sets of transformations of the \Fqt are strictly larger than the ones
considered in Ref.~\citen{Bravyi2002210}. However, a parity non preserving map \( \Transformation{T} \) on \(
N \) {\Lfm}s---\ie \( \Transformation{T} \) has Kraus operators that are linear combinations of products of an
odd number of field operators---can always be seen as the sequential composition \( \Transformation{X}_i \circ
{\Transformation{X}_i}^{-1} \circ \Transformation{T} \) acting on \( N \) {\Lfm}s, where \(
\Transformation{X}_i \) is the unitary map that flips the $i$th \Lfm from occupied to unoccupied, and vice
versa.  Notice that \( {\Transformation{X}_i}^{-1} \circ \Transformation{T} \) is now parity preserving. We
conclude therefore that in the \Fqt a non parity preserving map can be seen as the sequential composition \(
\Transformation{X}_i \circ \Transformation{R} \) of a parity preserving map \( \Transformation{R} \) and a
local flip \( \Transformation{X}_i \).

\subsubsection{Universality of computation}

We want to prove that in the \Fqt there is a finite set of Fermionic gates that allows us to build every \Fqt
circuit. Given a system of $N$ {\Lfm}s, there are the parity-preserving transformations and the parity
changing ones, which can be written as the sequential composition \( \Transformation{X}_i \circ
\Transformation{R} \), with \( \Transformation{R} \) parity-preserving. In Ref.~\citen{Bravyi2002210} a
universal set \( \Upsilon \) of \Lfm gates for the parity-preserving transformations is given; then it follows
that a universal set for the \Fqt is given by \( \Upsilon \cup \{ \Transformation{X}_i \} \) for some $i \in
\{ 1,\ldots,N \}$.

Let us now review the derivation of the universal set for parity preserving transformations.  The
proof\cite{Bravyi2002210} relies on the universality of computation in \Qt, and on the possibility of
expressing every parity-preserving \Lfm gate by means of qubit gates.  It is important to notice some
differences between the qubit computation and the Fermionic one. Consider a gate $\Transformation{G}$ acting
on $M$ qubits.  Such a gate is represented by a unitary operator $G$ acting on \( ({\Complexes^{2}})^{\otimes
M} \). When such a gate is used in a quantum circuit of $N>M$ qubits, its operator representative is always
given by the unitary operator $G \otimes I$, modulo a relabeling of the subsystems; more precisely since the
Hilbert space of $N$ qubits $ (\Complexes^{2})^{\otimes N}$ can be identified with $(\Complexes^{2})^{\otimes
M} \otimes (\Complexes^{2})^{\otimes (N-M)}$ by the qubit permutation $P:\,\Ket{s_1,\ldots,s_{N}}_\sQ \mapsto
\Ket{s_{j_1},\ldots,s_{j_{M}}}_\sQ \otimes \Ket{s_{j_{M+1}},\ldots,s_{j_N}}_\sQ$, the action of \(
\Transformation{G} \) on the qubits \( s_{j_1}, \ldots s_{j_M} \) is given by the operator
\begin{equation}\label{eq:pipperino}
  \tilde{G}(j_1, \ldots, j_M) := P^{-1}\ G \otimes I\ P.
\end{equation}
Clearly such a property is of paramount importance for the universality of computation, since the gate \(
\Transformation{G} \) is ``always'' represented by the operator $G$ irrespective of the number of the qubits
of the whole circuit, and irrespective of the specific choice of the qubits the gate acts on. 

In the \Lfm scenario the situation is very different due to the \Car.  For instance, a $2$-\Lfm gate
$\Transformation{F}$ behaves differently depending on the \Lfm subsystems $\Transformation{F}$ it acts on. For
example, let \( \F[1]^\dagger \F[2] \) be a parity-preserving Fermionic operator; when it is applied to the
{\Lfm}s \( j_1 \), \( j_2 \) of a multipartite system of \( N \) {\Lfm}s it behaves differently depending on
the chosen ordering for the $N$ subsystems, since
\begin{multline*} 
  \F[j_1]^\dagger \F[j_2] \Ket{\ldots,s_{j_1},\ldots,s_{j_2},\ldots}_\sF = \\
    \delta_{s_{j_1},0} \times \delta_{s_{j_2},1} \times (-1)^{\sum_{k=j_1+1}^{j_2-1}s_k}\times
    \Ket{\ldots,1,s_{j_1+1},\ldots,s_{j_2-1},0,\ldots}_\sF.
\end{multline*}
When we represent the {\Lfm} gate $\F[j_1]^\dagger \F[j_2]$ by means of qubits, such a difference in behaviour
is taken into account by the Jordan-Wigner transform thanks to the \( \sigma^z_k \) operators at the qubit
subsystems ranging from \( k=j_1+1 \) to \( k=j_2-1 \). This fact has the following consequence: a \Lfm
operator has many qubit representations according to the total number of {\Lfm}s involved.  However, whenever
\( j_1 \) and \( j_2 \) are nearest neighbours there is no contribution from the coefficient \(
(-1)^{\sum_{k=j_1+1}^{j_2-1}s_k} \) and every 2-\Lfm parity-preserving gate acting on nearest neighbour
{\Lfm}s admits an unambiguous qubit representation made of parity-preserving qubit gates (by means of the
\Jwt).  The same result holds for one-\Lfm parity-preserving transformations---due to the fact that
parity-preserving transformations are linear combinations of products of an even number of field operators.

This allows us to represent an arbitrary $2$-\Lfm gate $ \Transformation{T}( j,k ) $ acting on the {\Lfm}s
$j$, $k$ (w.l.o.g.~$j<k$) by means of qubits in an unambiguous way. Indeed, let us call by \( \fswap(j,j+1) \)
the Kraus operator of the unitary transformation performing the swap between the \( j \)th and the \( (j+1)
\)th \Lfm, \ie \( \fswap(j,j+1) \,\F[j]\,{\fswap(j,j+1)}^\dagger = \F[j+1] \), and \( \fswap(j,j+1)
\,\F[j+1]\,{\fswap(j,j+1)}^\dagger\, =\, \F[j] \), namely \( \fswap(j,j+1) = I - \F[j]^\dag \F[j] -
\F[j+1]^\dag \F[j+1] + \F[j+1]^\dag \F[j] + \F[j]^\dag \F[j+1] \). Such an operator acts in the following way:
\( \fswap(j,j+1) \Ket{\ldots,s_j,s_{j+1},\ldots}_\sF = (-1)^{j(j+1)} \Ket{\ldots, s_{j+1} ,s_{j},\ldots}_\sF
\).  Since the swap between the $j$th and the $(j+1)$th qubit of a circuit is given by the swap operator
$\qswap(j,j+1):\,\Ket{\ldots,s_j,s_{j+1},\ldots}_\sQ \mapsto \Ket{ \ldots, s_{j+1}, s_j, \ldots}_\sQ$, we have
that
\begin{equation*}
  J(\fswap(j,j+1)) = \qswap(j,j+1) D(j,j+1),
\end{equation*}
where $D(j,j+1):\ \Ket{\ldots,s_j,s_{j+1},\ldots}_\sQ \mapsto (-1)^{j(j+1)} \Ket{ \ldots, s_{j}, s_{j+1},
\ldots}_\sQ$ is the so-called \emph{swap defect operator}. Notice that also the swap defect operator $D$ is
parity-preserving and nearest-neighbour. Since
\begin{equation}\label{eq:asddsaasd}
  \Transformation{T}( j,k ) \equiv {\fswap(k-1,k)} \ldots {\fswap(j+1,j+2)} \Transformation{T}(j,j+1)
    \fswap(j+1,j+2) \ldots \fswap(k-1,k), 
\end{equation}
we have that
\begin{multline}\label{eq:qweewqqwe}
  J(\Transformation{T}( j,k )) = D(k-1,k) \ldots D(j+1,k) \\ 
  \qswap(k-1,k) \ldots \qswap(j+1,j+2) J(\Transformation{T}( j,j+1 )) \qswap(j+1,j+2) \ldots
    \qswap(k-1,k) \\ 
  D(j+1,k) \ldots D(k-1,k).
\end{multline}
Hence we have found that an arbitrary $2$-\Lfm parity-preserving operator is equivalent to a Fermionic circuit
involving only gates on nearest neighbour {\Lfm}s (Eq.~\eqref{eq:asddsaasd}), which can therefore be
represented unambiguously by the parity-preserving qubit circuit of Eq.~\eqref{eq:qweewqqwe}.  This method
works also for operators which act on more than two {\Lfm}s.  Notice that the term \( \qswap(j+1,j+2) \ldots
\qswap(k-1,k) \) in Eq.\eqref{eq:qweewqqwe} is just the permutation $P$ of Eq.~\eqref{eq:pipperino}.

Due to the equivalence between parity-preserving Fermionic gates and parity-preserving qubit gates we only
need a universal set of parity-preserving qubit gates in order to get a universal set of Fermionic
parity-preserving unitary transformations. A universal set for the qubits is given by\cite{Bravyi2002210}
\begin{equation}\label{eq:pp-univ}
  \Lambda(e^{i\pi/4}),\qquad \Lambda(\sigma^z) \equiv D,\qquad \tilde{H}:\ \Ket{a,b}_\sQ \mapsto
    \frac{1}{\sqrt{2}}\sum_{c}(-1)^{bc}\Ket{a \oplus b \oplus c,c}_\sQ,
\end{equation}
where \( \Lambda(U) \) denotes the controlled $U$ with the control system corresponding to the first qubit.  

The proof of universality of gates in Eq.~\eqref{eq:pp-univ} proceeds as follows (i) it is observed that any
parity-preserving qubit gate $U$ can be considered as a block-diagonal operators 
\begin{equation}\label{eq:parpres}
  U=\left(
  \begin{array}{c|c}
    W_0&0\\
    \hline 0&W_1
  \end{array}\right),
\end{equation}
where $W_i$ acts on the parity sector $\Hilb{H}_i=(\Complexes^2)^{\otimes N-1}$ of the Hilbert space \(
(\Complexes^2)^{\otimes N} \); (ii) it is shown how to get any parity-preserving operator having \( W_0=W_1
\); (iii) the operators having $W_0=I$ and $W_1= Y$, which transform the operators having $W_0=W_1$ to the
general form of Eq.~\eqref{eq:parpres}, are constructed.

Notice that any operator $G$ on $M-1$ qubits can be turned into a parity-preserving one $\ppext{G}$ on $M$
qubits by using an ancillary qubit:
\begin{multline}\label{eq:homo}
  \ppext{G} =\, V_M\,(I\otimes G)\,V_M,\\
  V_M^{-1}=V_M:\ \Ket{s_1,\ldots,s_{M}}_\sQ \mapsto \Ket{s_1\oplus\ldots\oplus s_{M},\,s_2,\ldots,s_{M}}_\sQ.
\end{multline}
Indeed, the unitary operator $V_M$ maps the parity sector $\ket j\otimes\Hilb{H}_i$ of $M$ qubits onto the
subspace $\Ket{i\oplus j}\otimes\Hilb{H}_i$, and then \( \ppext{G} \) is parity preserving, even if \( G \) is
not. Notice that if $G$ already preserves the parity then $\ppext{G}=I \otimes G$. This is the case of the
first two operators $\Lambda(e^{i\pi/4})$ and $\Lambda(\sigma^z)$ in Eq.~\eqref{eq:pp-univ}, while the last
universal gate $\ppext{H}$ is the parity-preserving extension of the usual Hadamard gate $H$.  Since for every
unitary $G$ on $N-1$ qubits the unitary $\ppext{G}$ of Eq.~\eqref{eq:homo} is parity-preserving, we have that
\( \ppext{G}\) is of the form of Eq.~\eqref{eq:parpres}. On the other hand, one can easily check that
$V_M(\sigma^x\otimes I_{M-1})V_M=\sigma^x\otimes I_{M-1}$, hence \( [ \sigma^x \otimes I_{M-1},\ppext{G} ] = 0
\).  Moreover, since $(\sigma^x\otimes I_{M-1})\ket i\otimes\ket{\xi_j}=\ket {i\oplus1}\otimes\ket{\xi_j}$, if
we identify the bases $\ket {0}\otimes\ket{s1,\ldots,s_{M-1}}$ and $\ket {1}\otimes\ket{s1,\ldots,s_{M-1}}$ in
the subspaces $\ket i\otimes\Hilb{H}_i$, we have
\begin{equation*}
  \sigma^x\otimes I_{M-1}=\left(
  \begin{array}{c|c}
    0&I\\
    \hline    I&0
  \end{array}\right),
\end{equation*}
which implies that for any $M-1$-qubits gate $G$ the parity-preserving extension $\ppext G$ has $W_0=W_1$.

Since Eq.~\eqref{eq:homo} defines a $*$-algebra homomorphism, any universal set of gates is mapped to a set of
parity-preserving gates that is universal on the even sector. The set of gates \( \{
\Lambda(e^{i\pi/4}),\,\Lambda(\sigma^z),\,H \} \) is known to be universal, then the corresponding
parity-preserving set given by Eq.~\eqref{eq:pp-univ} must be universal. Notice that the homomorphism
\eqref{eq:homo} satisfies the property: \( \ppext{\Lambda(X)} = {S^{c,p}}\, \Lambda(\ppext{X})\, S^{c,p} \)
where \( S^{c,p} \) is a swap between the control and the parity qubits.

We conclude that the set of Eq.~\eqref{eq:pp-univ} is universal for parity-preserving unitary gates having \(
W_0=W_1 \). We can use the same set to build parity-preserving unitary operators $K$ with $W_0=I$ and $W_1=Y$
to correct the first step. We add one ancillary qubit at the end of our $M$ qubits.  Let us define the
operator
\begin{equation*}
  Z:\ \Ket{s_1,\ldots,s_M,s_{M+1}}_\sQ \mapsto \Ket{s_1 \oplus \ldots \oplus s_M, s_2 ,\ldots,s_{M}, s_2
    \oplus \ldots \oplus s_{M+1}}_\sQ.
\end{equation*}
Let $K$ be a parity-preserving unitary operator with $W_0=I$ and $W_1=Y$, and let $\ppext{H}$ have diagonal
blocks $W_0=W_1=Y$.  Denoting by \( P \) the permutation $\Ket{s_1,s_2,\ldots,s_M,s_{M+1}}_\sQ \mapsto
\Ket{s_1,s_{M+1},s_2,\ldots,s_M}_\sQ$, we have
\begin{equation}
  Z^{-1}\,P^{-1} (\Lambda(\ppext{H}) \otimes I) P \, Z = (\,V^{-1} \Lambda(H) V\,) \otimes I = K\otimes I.
\end{equation}
We just need to represent the operator $Z$ by means of the operators in the universal set of
Eq.~\eqref{eq:pp-univ}. This task can be easily accomplished upon noticing that \( Z \equiv
\Lambda(\ppext{\sigma^x})(m-1,0,m) \Lambda(\ppext{\sigma^x})(m-2,0,m) \cdots
\Lambda(\ppext{\sigma^x})(1,0,m)\), where \( \Lambda(\ppext{\sigma^x})(j_1,j_2,j_3):\
\Ket{\ldots,s_{j_1},\ldots, s_{j_2},\ldots,s_{j_3},\ldots}_\sQ \mapsto \Ket{\ldots,s_{j_1},\ldots, s_{j_2}
\oplus s_{j_1},\ldots,s_{j_3} \oplus s_{j_1},\ldots}_\sQ \). Now, the operator \( \Lambda(\ppext{\sigma^x}) \)
acting on the qubits \( \sA \), \( \sB \), and \( \sC \) can be expressed in terms of the universal set as \(
\Lambda(\ppext{\sigma^x})_{\sA\sB\sC} = (\ppext{H}_{\sA\sC} \otimes I_\sB) \ ( \Lambda(\sigma^z)_{\sB\sC}
\otimes I_\sA ) \  ( \ppext{H}_{\sA\sC} \otimes I_\sB ) \).

Now we just need to represent the gates of Eq.~\eqref{eq:pp-univ} in terms of the creation and the
annihilation operators. The first two operators are
\begin{equation}\label{eq:lfm-univ}
  \Lambda(e^{i\pi/4})=\exp(i\frac{\pi}{4}\F[0]^\dagger \F[0]),\qquad
    \Lambda(\sigma^z)=\exp(i\pi \F[0]^\dagger \F[0]\F[1]^\dagger \F[1]).
\end{equation}
The gate $\ppext{H}$ can be represented in the {\Lfm} case by means of the decomposition
\begin{equation*}
\ppext{H}=\
  [ I \otimes \Lambda(-i) ] \cdot \ppext{G} \cdot [ I \otimes \Lambda(-i)],\qquad\ G=\left(
    \begin{array}{cc}
      1 & i \\ 
      i & 1 
    \end{array} \right) .
\end{equation*}
Hence a universal set for the parity-preserving gates of the \Fqt is given by the gates of
Eq.~\eqref{eq:lfm-univ} together with
\begin{equation*}
  \ppext{G}\,=\,
  \exp[-i\frac{\pi}{4}(\F[0]-\F[0]^\dagger)(\F[1]+\F[1]^\dagger)] \,=\,
  \exp[i\frac{\pi}{4}(\F[0]^\dagger \F[1]+\F[1]^\dagger \F[0])]
  \exp[i\frac{\pi}{4}(\F[1] \F[0]+\F[0]^\dagger \F[1]^\dagger)].
\end{equation*}

\subsubsection{Simulation}

We now address the issue of simulating a qubit circuit by means of a \Fqt circuit, and vice versa. The proof
of the universality given in the previous section gives already a way to simulate a \Lfm circuit by means of
qubits, relying on the Jordan-Wigner isomorphism between the qubit algebra and the Fermionic one. Moreover,
thanks to Section~\ref{sec:locc} we know that a Fermionic \Locc can be simulated by a \Locc on qubits. In
order to address the simulation in the other way round, and to tight the simulation cost of the previous
section, we will present the scheme of Ref.~\citen{Bravyi2002210} which does not rely on the identification \(
\Ket{s_1, \ldots, s_N}_\sF \leftrightarrow \Ket{s_1, \ldots, s_N}_\sQ \)---\ie the \Jwt.  As we will see, this
time the scheme of Ref.~\citen{Bravyi2002210} works out of the box even for the \Fqt.  In the following, for
the sake of convenience, we will label the \Lfm and the qubit systems starting from ``zero'', and not from
``one'' as we did in the rest of the paper.

Given a circuit of the \Fqt, a procedure to simulate a $K$-\Lfm gate \( \Transformation{T} \) can be
summarized as: (i) we embed the $M$ {\Lfm}s system in a $M$ qubits system, (ii) we add $K$ ancillary qubits
initialized in the state \( \Ket{0,\ldots,0}_\sQ \), (iii) we exchange the qubits corresponding to the {\Lfm}s
involved in the computation with the ancillas, taking into account possible global phases due to the
anticommutation relation of the original \Lfm systems (iv) we perform the computation on the ancilla by means
of the corresponding qubit gate, (v) we revert the extracted qubits in their original position, (vi) we
re-encode the resulting qubits---excluding the ancillary qubits---in the original $M$ {\Lfm}s.

Clearly, one possible way of encoding is given by the \Jwt, namely \( \Ket{s_0, \ldots, s_{M-1}}_\sF
\leftrightarrow \Ket{s_0, \ldots, s_{M-1}}_\sQ \). This is actually the same encoding used in the previous
section in order to derive the universal set for the Fermionic computation. In such a case the process of
embedding and of extraction (of the $j$th \Lfm) is synthetically given by
\begin{multline}\label{eq:JW-embedding}
  \Ket{s_0, \ldots, s_j, \ldots, s_{M-1}}_\sF \to \Ket{s_0, \ldots, s_j, \ldots, s_{M-1}}_\sQ \to \Ket{0, s_0,
    \ldots, s_j, \ldots, s_{M-1}}_\sQ \to \\ 
  \to (-1)^{s_j \oplus_{i=0}^{j-1} s_i} \Ket{s_j, s_0, \ldots, 0, \ldots, s_{M-1}}_\sQ.
\end{multline}

A simulation scheme resorting to the above \Jwt encoding is not very efficient, since every time we perform
the extraction of one qubit we shall evaluate a phase given by the coefficient \( (-1)^{s_j \oplus_{i=0}^{j-1}
s_i} \) of Eq.~\eqref{eq:JW-embedding}. Therefore, in the worst case scenario, for every $1$-\Lfm gate we have
to use $O(M)$ qubit gates. We can do better using a different embedding.

Let us introduce a partial ordering in the space of the binary strings: we say that the binary string \(
\alpha:=\alpha_{T-1}\ldots\alpha_0 \) precedes \(\beta:= \beta_{T-1}\ldots\beta_0 \), and we write \( \alpha
\preceq \beta\), whenever for some \( 0 \leq l_0 \leq T-1 \) we have \( \alpha_l = \beta_l \) for \( l \geq
l_0\), and \( \beta_{l} = 1 \) for \( l < l_0 \).  If we denote with \( j_\text{bin} \) the binary string
corresponding to the decimal number \( j \), we have \( j_\text{bin} \prec k_\text{bin} \implies j < k \),
where clearly \( j_\text{bin} \prec k_\text{bin} \Leftrightarrow ( j_\text{bin} \preceq k_\text{bin} ) \wedge
( j_\text{bin} \neq k_\text{bin} ) \).  Notice that given a binary string $ j_\text{bin} = \alpha$ of length
$T$, there are at most $T$ binary strings \( k_\text{bin} \) of the same length satisfying the relation
$j_\text{bin} \preceq k_\text{bin}$. Indeed, since for every \( 0 \leq l_0 \leq T-1 \) the strings greater
than or equal to \( j_\text{bin} \) are precisely those of the form \( \alpha_{T-1}\ldots\alpha_{l_0}1\ldots1
\) , and since there are at most $T$ different strings of this kind, an upper bound to the number of binary
strings greater than a given one is given by the string length $T$.

We can now consider the following encoding scheme:
\begin{multline}\label{eq:magical-encoding}
  \Ket{s_0, \ldots, s_{M-1}}_\sF \mapsto \Ket{x_0, \ldots, x_{M-1}}_\sQ, \\
  \text{where } x_j = \bigoplus_{ i \in S(j)} s_i, \text{ and } S(j) := \{ k \in \{ 0, \ldots, M-1 \} \mid
    k_\text{bin} \preceq j_\text{bin}  \}.
\end{multline}
It is very important to notice that since a $s_i$ appears in every \( x_j \) satisfying \( i_\text{bin}
\preceq j_\text{bin} \), we have that a $s_j$ appears at most in \( T \) of the $x_j$, with $T=
\operatorname{ceiling}(\log_2 M)$ being the number of bits required to binary encode the labels of the \Lfm
systems, ranging from \( 0 \) to \( M-1 \).

Let $\chi^{[l_0]}$ be the string having $\chi^{[l_0]}_l=1$ for $l \leq l_0$ and $\chi^{[l_0]}_l=0$ otherwise,
then the following two properties hold:
\begin{enumerate}[(i)]
\item the inversion of the relation \( x_j = \oplus_{i\in S(j)} s_i \) leads to \( s_j = x_j -
  \oplus_{i \in K(j)} x_i \), where for \( j_\text{bin} = \beta_{T-1}\ldots\beta_0 \)
  \begin{equation*}
    K(j) = \{ \alpha \mid \exists\ 0 \leq l_0 \leq T-1 \text{ s.t. } 
    \chi^{[l_0]}_l \beta_l =\chi^{[l_0]}_l, \  
    \alpha_l = \beta_l\oplus\delta_{ll_0}
    \};
  \end{equation*}
\item the quantity \( \oplus_{i=0}^{j-1} s_i \) can be written in terms of the encoded numbers \( x_l \) as \(
  \oplus_{i=0}^{j-1} s_i = \oplus_{i\in L(j)} x_i \) where \( j_\text{bin} = \beta_{T-1}\ldots\beta_0 \)
  \begin{equation*}
    L(j) = \{ \alpha \mid \exists\ 0 \leq l_0 \leq T-1 \text{ s.t. } \alpha_l =
    \beta_l\oplus\chi^{[l_0]}_l(\beta_l\oplus1),\ \beta_{l_0}(\alpha_{l_0} \oplus1) = 1\}.
  \end{equation*}
\end{enumerate}
Observe that also the sums appearing in the two above expressions contain at most \( \log_2
M \) elements.

While the extraction procedure of the qubits with the standard encoding given by the \Jwt requires a number of
computational steps linear in the number of the {\Lfm}s of the circuit, with this last encoding the situation
is improved: suppose to extract the $j$th qubit starting from the initial state \( \Ket{s_0, \ldots,
s_{M-1}}_\sF \) encoded in the qubit state \( \Ket{x_0, \ldots, x_{M-1}}_\sQ \). First of all we add the
ancillary qubit \( \Ket{0} \) at the beginning of the string (let us call it ``the qubit at the position
$-1$''), then the extraction goes as follows
\begin{multline*}
 \Ket{0, x_0, \ldots, x_{M-1}}_\sQ \overset{A}{\rightarrow}  \Ket{s_j, x_0, \ldots, x_{M-1}}_\sQ
   \overset{B}{\rightarrow} \Ket{s_j, x_0^\prime, \ldots, x_{M-1}^\prime}_\sQ \overset{C}{\rightarrow} \\
 \overset{C}{\rightarrow} (-1)^{s_j \oplus_{i=0}^{j-1}s_i} \Ket{s_j, x_0^\prime, \ldots, x_{M-1}^\prime}_\sQ,
\end{multline*}
where
\begin{enumerate}[A]

  \item is a unitary evolution  that evaluates the value of \( s_j \) from the encoded string \( x_{M-1} \ldots
    x_0 \) and writes it into the ancillary qubit. Such an operation can be achieved by means of the unitary map
    \begin{equation}
      A = \prod_{i\in K(j)\cup\{j\}} \Lambda(\sigma^x)(i,-1),
    \end{equation}
    where we remember that \( \Lambda(U)(i_0,\ldots,i_p) \) represents the controlled unitary \( U \) with
    control system \( i_0 \) and target systems \( i_1,\ldots,i_p \). Since the cardinality of \( K(j) \) is \(
    O(\log M) \), we will need \( O(\log M ) \) gates to perform $A$;

  \item turns the original \( s_j \) (not the copy in the ancillary qubit) to zero. The transformation $B$
    then must change the encoded string in such a way that the following diagram commutes:
    \begin{equation*}
      \begin{array}[c]{ccc}
        \Ket{s_0, \ldots, s_j, \ldots, s_{M-1}}_\sQ & \longrightarrow & \Ket{s_0, \ldots, s_j\oplus s_0, \ldots,
          s_{M-1}}_\sQ \\ 
        \downarrow{\scriptstyle E} & & \downarrow{\scriptstyle E} \\
        \Ket{x_0, \ldots, x_j, \ldots, x_{M-1}}_\sQ & \overset{B}{\longrightarrow} & \Ket{x^\prime_0, \ldots,
          x^\prime_j, \ldots, x^\prime_{M-1}}_\sQ 
      \end{array}.
    \end{equation*}
    This operation is achieved by the following unitary
    \begin{equation*}
      B = \prod_{i_\text{bin} \succeq j_\text{bin}} \Lambda(\sigma^x)(-1,i),
    \end{equation*}
    where again the number of gates required is \( O(\log M) \);

  \item is an unitary evolution that evaluates the phase due the exchange of the Fermionic wires:
    \begin{equation*}
      C = \prod_{j \in L(j)} \Lambda(\sigma^z)(-1,i).
    \end{equation*}
    Again, since \( | L(j) | \approx O( \log M ) \), the number of the required gates amounts to \( O( \log M
    ) \).

\end{enumerate}
In conclusion to simulate a 1-\Lfm gate in a circuit of $M$ {\Lfm}s by means of qubit gates we need \( O( \log
M ) \) qubit gates (instead of \( O( M ) \) gates needed using the \Jwt encoding). This result holds for every
\Lfm operator---\ie for a $K$-\Lfm gate (with $K\leq M$) one needs to extract $K$ qubits by means of the above
procedure. Moreover notice that the proof does not require the gates to be parity-preserving.  Indeed the
reviewed procedure of Ref.~\citen{Bravyi2002210} provides an efficient way to perform the qubit extraction (or
equivalently to take into account the phase factor given by the $\sigma^z$ of the \Jwt) irrespective of the
parity features of the gate we want to simulate. Hence, the 1-\Lfm transformation $\Transformation{X}(\rho) =
(\F[i]^\dagger + \F[i]) \rho (\F[i]^\dagger + \F[i]) $, which is parity changing, can also be achieved by
means of \( \log M \) qubit gates: we just need to perform a $\sigma^x$ on an extracted qubit.

As shown in \citen{Bravyi2002210}, an efficient simulation of a $N$-qubit circuit by means of a \Lfm circuit
is easier. First of all one performs the encoding of the $N$ qubits into $2N$ qubits through the isometric
embedding \( V:\ \Ket{ s_0, \ldots, s_{N-1} }_\sQ \mapsto \Ket{ s_0, s_0, \ldots, s_{N-1}, s_{N-1} }_\sQ \). A
quantum gate $\Transformation{G}$ acting on the $j$th and the $k$th qubits---thus represented by the unitary
operator $\tilde{G}(j,k)$ of Eq.~\eqref{eq:pipperino}---is represented on the $2N$ qubits by the gate \(
\tilde{G}^\prime( 2j, 2j+1, 2k, 2k+1 ) = V \tilde{G}(j,k) V^\dagger \). If we embed the resulting $2N$-qubit
circuit into $2N$ {\Lfm}s by means of the \Jwt, the resulting $J(\tilde{G}^\prime(j,j+1,k,k+1))$, besides
being parity-preserving, is also made of field operators acting only on the {\Lfm}s \( 2j \), \( 2j+1 \), \(
2k \), \( 2k+1 \), namely no field operators on the rest of the circuit are needed. In conclusion, every qubit
gate acting on $2$ qubits can be simulated by means of a $4$-\Lfm gate. The same result clearly generalizes
for gates with an arbitrary number of qubits.

\section{Acknowledgments}
This work has been supported in part by the Templeton Foundation under the project ID\# 43796 {\em A
Quantum-Digital Universe}.  


\end{document}